\documentclass[journal,onecolumn]{IEEEtran}

\newif\ifhighlightchanges
\highlightchangesfalse 
\ifhighlightchanges
  \usepackage[draft]{changes}
\else
  \usepackage[final]{changes}
\fi

\usepackage{xcolor}
\usepackage{amsthm}
\usepackage{amsfonts}
\usepackage{thmtools}
\usepackage{thm-restate}
\usepackage{xcolor}
\usepackage{amssymb}
\usepackage{amsmath}
\usepackage{mathtools}
\usepackage{verbatim}
\usepackage{physics}
\usepackage{cite}
\usepackage{tablefootnote}
\usepackage[ruled,linesnumbered]{algorithm2e}
\newtheorem{theorem}{Theorem}
\newtheorem{example}{Example}
\newtheorem{lemma}[theorem]{Lemma}
\newtheorem{definition}[theorem]{Definition}
\newtheorem{corollary}[theorem]{Corollary}
\newtheorem{construction}{Construction}

\newtheorem{fact}[theorem]{Fact}
\newtheorem{remark}[theorem]{Remark}

\newcommand{\ceil}[1]{\left\lceil #1 \right\rceil}
\newcommand{\floor}[1]{\left\lfloor #1 \right\rfloor}
\usepackage{graphicx}
\usepackage{tikz}
\usetikzlibrary{patterns}
\usetikzlibrary{fit,backgrounds}
\usetikzlibrary{shapes}
\usepackage{pgfplots}
\usepackage{msc}
\DeclareMathOperator{\poly}{poly}
\DeclareMathOperator{\nei}{N}
\DeclareMathOperator{\Dec}{Dec}

\DeclareMathOperator{\Rep}{Rep}
\newcommand{\alg}[1]{\textsc{#1}\xspace}
\newcommand{\cJ}{\mathcal{J}}
\newcommand{\cC}{\mathcal{C}}

\usepackage{cleveref}

\newcommand{\cE}{\mathcal{E}}

\begin{document}

\title{Constructing Low-Redundancy Codes via Distributed Graph Coloring}

\author{%
Yuting~Li,  Ryan~Gabrys,~\IEEEmembership{Member,~IEEE,} Farzad~Farnoud,~\IEEEmembership{Member,~IEEE}%
\thanks{Yuting Li is with the Department of Computer Science at the University of Virginia, USA, \texttt{mzy8rp@virginia.edu}.}
\thanks{Ryan Gabrys is with Calit2 at the University of California-San Diego, USA, \texttt{rgabrys@ucsd.edu}.}
\thanks{Farzad Farnoud is with the Department of Electrical and Computer Engineering and the Department of Computer Science at the University of Virginia, USA, \texttt{farzad@virginia.edu}.}
\thanks{This paper was presented in part at ISIT 2023, Taipei, Taiwan.}
\thanks{This work was supported in part by NSF grants CIF-2312871, CIF-2312873, and CIF-2144974.}
}

\maketitle

\begin{abstract}
We present a general framework for constructing error-correcting codes using distributed graph coloring under the LOCAL model. Building on the correspondence between independent sets in the confusion graph and valid codes, we show that the color of a single vertex - consistent with a global proper coloring - can be computed in polynomial time using a modified version of Linial’s coloring algorithm, leading to efficient encoding and decoding. Our results include: i) uniquely decodable code constructions for a constant number of errors of any type with redundancy twice the Gilbert–Varshamov bound; ii) list-decodable codes via a proposed extension of graph coloring, namely, hypergraph labeling; iii) an incremental synchronization scheme with reduced average-case communication when the edit distance is not precisely known; and iv) the first asymptotically optimal codes (up to a factor of 8) for correcting bursts of unbounded-length edits. Compared to syndrome compression, our approach is more flexible and generalizable, does not rely on a good base code, and achieves improved redundancy across a range of parameters.
\end{abstract}

\section{Introduction}\label{sec:intro}

Certain types of errors, such as insertions and deletions, present unique challenges in the construction of error-correcting codes, even for a constant number of errors. These errors change the positions of the symbols in the input sequence, making some of the common approaches, such as algebraic coding, difficult to utilize. This challenge has led to disparate progress in the design of error-correcting codes for edit errors versus substitution and erasure errors. 
In particular, it took more than half a century after Levenshtein's optimal single-deletion-correcting codes~\cite{levenshtein1966binary} to construct codes with redundancy within a known constant factor of optimal~\cite{sima_two_2018,gabrys2018codes,guruswami2021explicit} to correct two deletions. 

In this work, we present a general framework for constructing error-correcting codes, including edit errors, with low redundancy and polynomial-time encoding and decoding based on distributed graph coloring. The redundancy of the resulting codes is twice the Gilbert-Varshamov bound (\added{Levenshtein bound in the edit distance case}). We accomplish this by translating insights from distributed graph coloring in the LOCAL model into our code-construction technique. The LOCAL model~\cite{linial87} is designed to study how local graph information can be transformed into a global solution, and how much do we lose by limiting our information to a specific radius. This perspective is useful for correcting codes for a constant number of errors. Although the encoding problem is global, as we will demonstrate, it is possible to solve it locally in a way that is consistent with the global solution with a small loss, when the number of errors is bounded. 

To construct codes in this framework, we will utilize Linial's coloring algorithm~\cite{linial87}. This algorithm, like other LOCAL algorithms, assumes an infinite amount of computational power at each vertex, as the focus is on the role of information in solving the problem. For our problem, however, computational complexity is critical, as our goal is to construct codes with polynomial encoding and decoding. We thus develop a modified version of the algorithm that is computationally efficient under mild conditions. 

To our knowledge, the only other general code construction method for a constant number of errors (for essentially any channel) is the celebrated work~\cite{sima2020syndrome},  introducing syndrome compression. We show that syndrome compression can be viewed in this framework as utilizing a specific distributed graph coloring method. Compared to the method proposed here, syndrome compression requires us to start with a reasonably good code, which may be difficult for certain channels. As a consequence, we are able to construct new codes with improved parameters, as detailed in Section~\ref{subsec:sc}.
 
The graph coloring framework is also more generalizable. In particular, we will present an alternative problem, called \emph{labeling} and show that distributed labeling algorithms can be used to construct list-decodable codes. Specifically, we present the first general construction method for list-decodable codes correcting a constant number of errors with a constant list size. Our technique hinges on specific combinatorial structures. We first prove they exist and then show that the existential proof  produces them with high probability. Using this approach, we construct codes with list size $\ell$ that can correct $k$ edits, with redundancy approximately $(3+1/\ell)k\log n$, which is strictly less than twice the GV bound. In the graph coloring framework, we also construct the first codes for correcting bursts of edits of unbounded length with redundancy approximately twice the GV bound. 

Finally, we construct codes for incremental synchronization. This problem arises in data synchronization when the two parties are uncertain about the edit distance between the two versions of the data. Here, assuming the largest possible distance will lead to inefficient synchronization. The proposed codes, constructed using the graph coloring framework, provide a lower expected communication load through an iterative encoding.

\subsection{Our Approach}\label{subsec:app}
We will now describe an overview of our approach, postponing the details to the following sections. In our framework, we will make use of the confusion  graph $G$ for a communication channel, introduced by Shannon~\cite{shannon_zero_1956}. The vertex set $V$ of this graph is the set of all possible input vectors of length $n$, e.g., $\{0,1\}^n$, and two vertices are adjacent precisely when they can produce the same channel output. It is well known and easy to see that an error-correcting code for the channel is equivalent to an \emph{independent set} in this graph, i.e., a set of vertices no two of which are adjacent. Now, if one finds a proper \(M\)-coloring
\[
\Phi\colon V \to \{1,\dots,M\},
\]
so that adjacent vertices receive distinct colors, then each color class
\[
\{x\in V : \Phi(x)=i\},\quad i=1,\dots,M,
\]
is an independent set, and hence a code. Among these, there exists a code with redundancy at most $\log M$.  However, computing such a coloring is rarely feasible in practice, as the number of vertices is exponential in $n$. Furthermore, even an efficiently computable \(\Phi(x)\) does not directly yield an efficient encoding algorithm. 

To overcome this difficulty, we propose to utilize the coloring in a different way. Consider the code 
\[
\cC = \{(x,\Phi(x)):x\in V\}.
\]
For the moment, suppose in each use of the channel, we send $x$ and $\Phi(x)$, but only $x$ is subject to errors (in accordance with $G$) and $\Phi(x)$ is received error-free. \added{We call these channels \emph{synchronization} channels. The motivation for this terminology is discussed in Remark~\ref{rem:err-free}.} Under this \replaced{channel}{model}, $\cC$ is an error-correcting code. \added{For general channels, protecting $\Phi(x)$ requires a negligible amount of redundancy as we discuss in Remark~\ref{rem:err-free} and the Appendix. Our focus on synchronization channels is merely to simplify the presentation and highlight the main points of the construction.} As this code is systematic, if one can efficiently evaluate $\Phi(x)$, then encoding is straightforward. For decoding, when $y$ is received, we first construct the set of all $x'$ that can produce $y$. Note that each element of this set is either equal to $x$ or one of its neighbors. So only one element in this set will have color $\Phi(x)$, thus allowing us to uniquely identify $x$.

We thus observe that for encoding, we only need to compute the color $\Phi$ of one vertex,  
and do not need to know the color of every vertex. This leads to the following question: Is there an algorithm/function $\Phi$ for coloring $G$ that is efficiently computable for a single vertex? There is, in fact, a class of algorithms with this property, namely, \emph{distributed graph coloring} algorithms. 

We now describe the LOCAL model for distributed graph algorithms, introduced by Linial~\cite{linial92}, and then we will present how LOCAL distributed graph coloring can be used for efficient code construction. A LOCAL algorithm is executed in rounds, where in each round, each node in the graph can perform unbounded computation and send messages of unbounded size to its neighbors. All communication is performed synchronously and reliably. Note that each vertex in this model has an infinite amount of computational power, as the original goal of the model is to study how local information can yield a global solution.

Consider now a LOCAL coloring algorithm $\alg{alg}$ such that the amount of computation performed by each vertex in each round would take time at most $T$ if it were to be performed on a ``regular'' processor. Furthermore, suppose that this algorithm takes $R$ rounds. At the end of the execution, each node $x$ will possess $\Phi(x)$, where $\Phi$ is a specific proper coloring of $G$.

Suppose now that $G$ is a confusion graph with $2^{\poly(n)}$ vertices (vectors), and we would like to encode $x$ as $(x,\Phi(x))$. As $\alg{alg}$ takes $R$ rounds, and in each round each node only communicates with its neighbors, the value of $\Phi(x)$ for any vertex $x$ is independent from any vertex $x'$ at distance larger than $R$. So it suffices to perform the computation only at nodes within a distance $R$ of $x$ (this statement will be made precise in the proof of \Cref{thm:rt}). This means that the total amount of computation to compute $\Phi(x)$ for a given $x$ on a regular processor is at most $O(\Delta^{O(R)}R^2T)$, where $\Delta$ is the maximum degree of $G$. Hence, $\alg{alg}$ yields an error-correcting code with encoding  complexity $O(\Delta^{O(R)}R^2T)$ (see \Cref{thm:rt}).

Linial~\cite{linial92} provides a LOCAL $O(\Delta^2)$-coloring algorithm that runs in $O(\log^* n)$ rounds on a graph with $2^{\poly n}$ vertices. Linial's algorithm is based on \textit{re-coloring} the vertices in each round. Namely, each vertex starts with its ID as its color (e.g., a binary sequence of length $n$). In each round, each vertex broadcasts its current color to all its neighbors. Then each vertex computes a new color from a smaller set in a way that it does not conflict with any of its neighbors. In particular, it is shown that given a proper coloring with $2^k$ colors, in one round, it is possible to obtain a proper coloring with $5k\Delta^2$ colors (although the proof with these specific constants in~\cite{linial92} is nonconstructive). 

Unfortunately, a direct application of Linial's algorithm leads to super-polynomial encoding complexity, as $R$ is not a constant and $\Delta$ for confusion graphs of interest is polynomial in $n$. Additionally, Linial's algorithm does not guarantee a polynomial $T$. Hence, the complexity $O(\Delta^{O(R)}R^2T)$ is super-polynomial. In other words, while this LOCAL algorithm is useful for us mainly through the limitation it puts on the dependence between the colors of different nodes, it does not limit the amount of computation. To remedy this, we will present a modified version of Linial's algorithm that in fact leads to an  efficient encoding and decoding (See \Cref{thm:GV}).

We have thus outlined a method for efficiently encoding data $x$ as $(x,\Phi(x))$ using distributed graph coloring under the LOCAL model, along with efficient decoding. The redundancy, i.e., the length of $\Phi(x)$, is $2\log\Delta+O(\log \log n)$ symbols. Note that according to the GV bound, there exists a code with redundancy $\log(\Delta+1)$. So the proposed codes achieve twice the redundancy given by the GV bound.

We extend this approach to list decoding (\Cref{thm:list}), utilize for incremental synchronization (\Cref{thm:incremental}), and construct codes that outperform the state of the art (\Cref{cor:kedit}, \Cref{thm:SSE}).

\begin{remark}\label{rem:err-free}
   Recall that we assumed $\Phi(x)$ can be transmitted between the encoder and decoder error-free. In practice, we will also need to protect the information $\Phi(x)$. However, if the length of $\Phi(x)$ is small, say $O(\log n)$, this can be accomplished with redundancy $O(\log \log n)$ for channels such as the deletion channel \added{using Lemma~\ref{lem:rep1} of the Appendix}. Additionally, for data synchronization, where $x$ is the updated document and $y$ is the original document, $\Phi(x)$ is indeed transmitted through a reliable link. For the sake of simplicity, in the rest of the paper, we will refer to a set of the form $\{(x,\Phi(x)):x\in\Sigma_2^n\}$ as an error-correcting code for a \added{synchronization channel}  if for any inputs $x,x'$ that can produce the same output $y$, we have $\Phi(x)\neq \Phi(x')$.
\end{remark}

\begin{figure}
    \centering
    \begin{tikzpicture}[
  scale=1,
  every node/.style={circle,draw=black,inner sep=1pt}]

\tikzset{vt0/.style={circle,draw=black,fill=yellow!70,text=black,postaction={pattern=horizontal lines,pattern color=gray!50}}}
\tikzset{vt1/.style={circle,draw=black,fill=cyan!50,text=black,postaction={pattern=vertical lines,pattern color=cyan!60!black}}}
\tikzset{vt2/.style={circle,draw=black,fill=red!70,text=white,postaction={pattern=north east lines,pattern color=white}}}
\tikzset{vt3/.style={circle,draw=black,fill=blue!80,text=white,postaction={pattern=north west lines,pattern color=white}}}

\node[vt0] (000) at (0.00,0.00) {000};
\node[vt1] (100) at (1,1) {100};
\node[vt3] (001) at (1,-1) {001};
\node[vt2] (010) at (1.5,0) {010};
\node[vt0] (101) at (3,0) {101};
\node[vt3] (110) at (3.5,1) {110};
\node[vt1] (011) at (3.5,-1) {011};
\node[vt2] (111) at (4.5,0) {111};

\draw (000) -- (001);
\draw (000) -- (010);
\draw (000) -- (100);
\draw (001) -- (010);
\draw (001) -- (011);
\draw (001) -- (100);
\draw (001) -- (101);
\draw (010) -- (011);
\draw (010) -- (100);
\draw (010) -- (101);
\draw (010) -- (110);
\draw (011) -- (101);
\draw (011) -- (110);
\draw (011) -- (111);
\draw (100) -- (101);
\draw (100) -- (110);
\draw (101) -- (110);
\draw (101) -- (111);
\draw (110) -- (111);
\end{tikzpicture}
\color{black}
    \caption{The confusion graph for the 1-deletion channel with input $\{0,1\}^3$, along with a proper coloring. The graph coloring shown here is such that each color class corresponds to a VT code.}
    \label{fig:confGraph}
\end{figure}
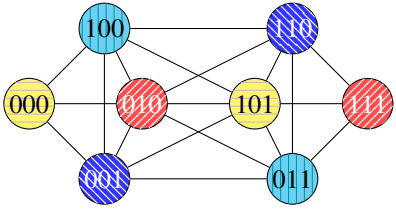

\subsection{Our Results}
In the framework described in the previous subsection, we study four problems; namely i) general code construction, ii) list decoding, iii) incremental synchronization, and iv) codes correcting substring edit errors. In this subsection, we present the main result for each problem.

  %For a string $x$, $\abs{x}$ denotes its length. For a set $S$, $\abs{S}$ denotes the cardinality of $S$. The redundancy of a code $\cC \subset \Sigma_2^n$ is $n-\log \abs{\cC}$ bits.
\subsubsection*{Uniquely Decodable Codes} Our first result pertains to the construction of codes for combinatorial channels and will be proven in Section~\ref{sec:wcerrors}. We represent a combinatorial communication/storage channel, or equivalently, the set of errors that can be applied to the channel input by $\cE$. Specifically, for $x\in\Sigma_q^n$, where $\Sigma_q$ is the set of $q$-ary strings of length $n$, we denote by $\cE(x)$ the set of all possible channel outputs for the input $x$. The confusion graph $G$ for a channel $\cE$ and set of inputs $\Sigma_q^n$ is a graph with vertex set $\Sigma_q^n$, with an edge between $x,x'\in\Sigma_q^n$ if they can both produce some $y$, i.e., $\cE(x)\cap\cE(x')\neq \emptyset$. An example is given in \Cref{fig:confGraph}. Denote by $\deg v$ the degree of a vertex $v$ in $G$ and by $\Delta(G)$ the maximum degree.

Our first result enables the use of LOCAL distributed coloring algorithm for construction of error-correcting codes (see \Cref{rem:err-free}).

\begin{restatable}{theorem}{thmrt}\label{thm:rt}
     Let $G_n$ be the confusion graph for a \added{synchronization} channel $\cE$ with inputs $\Sigma_q^n$. Suppose $\Delta_n$ is the maximum degree of $G_n$. Suppose a LOCAL algorithm takes $R_n$ rounds to compute a $B_n$-coloring $\Phi:\Sigma_q^n\to [B_n]$, and each round takes time at most $T_n$. Then, 
     \[
     \{(x,\Phi(x)): x\in\Sigma_q^n\}
     \]
     is an error-correcting code for $\cE$ with encoding complexity $O(\Delta_n^{R_n}(R_n+1)R_nT_n)$, decoding complexity \added{$O(\Delta_n^{R_n+1}(R_n+1)R_nT_n)$} and redundancy $\log B_n$.
\end{restatable}

To allow the algorithm to satisfy the required complexity upper bounds, we introduce the notion of explicitness.  A graph $G$ is $T$-\textit{explicit} if there is an ordering of vertices and an algorithm such that given a vertex index $v$ and $1\le i\le \deg v$, the algorithm can output $v$'s $i$-th neighbor in this ordering in time $T$. 

In \Cref{sec:wcerrors}, we will present a coloring algorithm that will help establish the following theorem.
 
\begin{restatable}{theorem}{thmGV}\label{thm:GV}
    Let $G_n$ be the confusion graph for a \added{synchronization} channel $\cE$ with inputs $\Sigma_q^n$, and suppose $\Delta_n = \Delta(G_n) = O(\poly(n))$, where $q,n$ are positive integers. If $G_n$ is $\poly(n)$-explicit, then one can construct a  (coloring) function $\Phi: \Sigma_q^n \to \Sigma_2^{\rho}$, where $\rho = 2 \log \Delta_n + O(\log \log n)$, such that 
    \[
        \{(x,\Phi(x)):x\in\Sigma_q^n\}
    \]
    is an error-correcting code for the channel $\cE$, with $\poly(n)$ encoding and decoding complexity and redundancy $2\log \Delta_n+O(\log\log n)$.
\end{restatable}

As an example of the implications of the previous theorem, suppose that $\cE$ is the $k$-deletion channel, i.e., up to $k$ input symbols may be deleted. Then the confusion graph $G_n$ for $\cE$ with input set $\Sigma_2^n$ has maximum degree $\Delta(G_n)=O(n^{2k})$. Applying Theorem~\ref{thm:GV}, yields a systematic code capable of correcting $k$ deletions with redundancy $4k\log n$. A generalization of this code  appears in a more rigorous manner in Corollary~\ref{cor:kedit}. Compared to syndrome compression, Corollary~\ref{cor:kedit} can be used to construct edit correcting codes over $q$-ary alphabet where $q=\poly(n)$, with a smaller redundancy than those in~\cite{sima2020optimalq}. Additionally, our constructions do not rely on existing codes.

\subsubsection*{List decoding} Our second result generalizes local graph coloring codes, and uses labelings of hypergraphs, a generalization of colorings, to construct list-decodable codes correcting a constant number of errors with constant list size.

For a channel $\cE$, let $\cE^{-1}(y)$ denote the set of all $x$ that can produce a string $y$. The confusion hypergraph $H_n$ for a channel $\cE$ with the set of possible inputs $\Sigma_q^n$ is a hypergraph whose vertex set is $\Sigma_q^n$ and edge set is $\{\cE^{-1}(y):y\in\Sigma_q^*\}$. An example is given in \Cref{fig:confHyperGraph}.

An $\ell$-labeling of a hypergraph, precisely defined in \Cref{def:lcoloring}, is a function $\Phi$ that assigns a label to each vertex of the hypergraph such that each edge has at most $\ell$ vertices of each label. A 2-labeling is given in \Cref{fig:confHyperGraph}. We observe that an $\ell$-labeling of $H_n$ provides a list-decodable code 
\[
\{(x,\Phi(x)):x\in \Sigma_q^n\}
\]
with list size $\ell$ for \added{the synchronization channel}  $\cE$. Specifically, for any $x\in\Sigma_q^n$ and $y\in\cE(x)$, we can find a list $L=\{x_1,\dotsc,x_\ell\}$ using $y$ and $\Phi(x)$ such that $x\in L$.

% \textbf{\begin{restatable}{theorem}{thmlrt}\label{thm:lrt}
%      Let $H_n$ be the confusion hypergraph for a channel $\cE$ with inputs $\Sigma_q^n$. If a LOCAL algorithm takes $R_n$ rounds to compute an $\ell_n$-labeling $\Phi:\Sigma_q^n\to [B_n]$, and each round takes time at most $T_n$, then for any vertex $x$, $\Phi(x)$ can be computed in time $(r_nv_n)^{R_n}(R_n+1)R_nT_n$, where $r_n$ is the maximum degree of $H_n$ and $v_n$ is the maximum edge size of $H_n$. Furthermore, if  $R_n$ is a constant and $T_n$, $r_n$, $v_n$ are polynomials in $n$, then
%      \[
%      \{(x,\Phi(x)): x\in\Sigma_q^n\}
%      \]
%      is a list-decodable code for $\cE$ with list size $\ell$ and polynomial-time encoding and decoding and redundancy $\log B_n$.
% \end{restatable}}

\begin{figure}
    \centering
\begin{tikzpicture}[
  scale=.66,
  every node/.style={draw=black,inner sep=1pt,fill=white,text=white},
  vtr/.style={circle,draw=black,fill=red!70,text=white,postaction={pattern=north east lines,pattern color=white}},  % “red” vertices
  vtb/.style={circle,draw=black,fill=blue!80,text=white,postaction={pattern=north west lines,pattern color=white}}, % “blue” vertices
  hyperedge/.style={
    draw=black,
    text=black,
    fill=none,           % no fill
    thick,
    ellipse,             % make the fit-node an ellipse
    inner sep=-1pt        % smaller inset for tighter fit
  }
]
  % vertices
  \def\c{1.5}                % <— tweak this

  % 2-coloring: red for parity(b1+b3)=0, blue for parity(b1+b3)=1
  \node[vtr] (000) at ({\c*0.0},{\c*0.0}) {000};
  \node[vtb] (100) at ({\c*1.0},{\c*1.0}) {100};
  \node[vtb] (001) at ({\c*1.0},{-\c*1.0}) {001};
  \node[vtr] (010) at ({\c*1.5},{\c*0.0}) {010};
  \node[vtr] (101) at ({\c*3.0},{\c*0.0}) {101};
  \node[vtb] (110) at ({\c*3.5},{\c*1.0}) {110};
  \node[vtb] (011) at ({\c*3.5},{-\c*1.0}) {011};
  \node[vtr] (111) at ({\c*4.5},{\c*0.0}) {111};

  % hyperedges as maximal cliques
  \begin{scope}[on background layer]
    \node[hyperedge,xscale=.9,label={[black]left:00}] (H1) [fit=(000)(001)(010)(100)] {};
    \node[hyperedge,xscale=.9,yscale=1.1,yshift=-.15cm,label={[black]below:01}] (H2) [fit=(001)(010)(011)(101)] {};
    \node[hyperedge,xscale=.9,yscale=1.1,yshift=.15cm,label={[black]above:10}] (H5) [fit=(010)(100)(101)(110)] {};
    \node[hyperedge,xscale=.9,label={[black]right:11}] (H6) [fit=(011)(101)(110)(111)] {};
  \end{scope}

\end{tikzpicture}
    \caption{The confusion hypergraph for the $1$-deletion channel with input $\{0,1\}^3$, along with a 2-labeling (red and blue). The edge labels, 00, 01, 10, and 11 are the output corresponding to the edge.}
    \label{fig:confHyperGraph}
\end{figure}
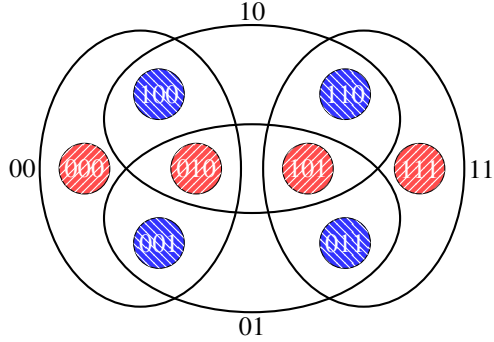

The problem of constructing $\ell$-labelings has not, to the best of our knowledge, been studied before. We will introduce an algorithm that provides an $\ell$-labeling when the hypergraph is explicit. 
For a hypergraph $H$, we use $r(H)$ to denote the maximum degree of $H$ and use $v(H)$ to denote the size of the maximum edge. The concept of explicitness for hypergraphs is a bit different from graphs and is defined precisely in Section~\ref{subsec:not}.

\begin{restatable}{theorem}{thmlist}\label{thm:list}
    Let $H_n$ be the confusion hypergraph for a \added{synchronization} channel $\cE$ with input set $\Sigma_q^n$. 
    Suppose $r(H_n)=O(n^b),v(H_n)=O(n^c)$ for constants $b,c$, and $H_n$ is $\poly(n)$-explicit. Then there exists a ($\ell$-labeling) function $\Phi: \Sigma_q^n \to \Sigma_2^{\rho}$
     where $\rho = (b(1+1/{\ell})+2c) \log n + O(\log \log n)$, such that 
    \[
        \{(x,\Phi(x)):x\in\Sigma_q^n\}
    \]
    is a list decodable code for the channel $\cE$, with list size $\ell$ and $\poly(n)$ encoding and decoding complexity.
    %, such that for any $x\in\Sigma_q^n,y\in\Sigma_q^*$ satisfying $y\in\cE(x)$, one can recover a list of $\ell$ containing $x$ from $y$ and $\Phi(x)$. Moreover, computing $\Phi$ and the recovery process both take time $\poly(n)$.
\end{restatable}
\added{To compare Theorem~\ref{thm:list} with Theorem~\ref{thm:GV} and see how $\Delta(G_n)$ relates to $r(H_n)$ and $v(H_n)$, we study the $k$-edits channel in Corollary~\ref{cor:keditl}, where $r(H_n)=O(n^k)$, $v(H_n)=O(n^k)$, and $\Delta(G_n)=O(n^{2k})$.
}
\paragraph*{Incremental Synchronization} Our third result is a file synchronization scheme, which we call \emph{incremental synchronization}. Consider the following scenario. Alice has a string $x$, and Bob has a string $y$ that differs from $x$ by a number of edits. Alice wants to send a syndrome $\Phi(x)$ to Bob so that 
Bob can recover $x$ given $y$ and $\Phi(x)$. If Alice knows $d_E(x,y)=k$, then Alice can send a syndrome $\Phi_k(x)$ by applying Theorem~\ref{thm:GV} with $\cE$ being $k$ edits and \added{$\Delta_n=n^{2k}$}. However, in general Alice does not know $k$. In this case, Alice can first send $\Phi_{k_1}(x)$ to Bob for some small $k_1$. If Bob can recover $x$, then Alice can stop here. If Bob cannot recover $x$, then Alice sends $\Phi_{k_2}(x)$ to Bob for some larger $k_2$. Alice can keep sending $\Phi_{t}(x)$ for increasing $t$ until Bob can recover $x$. Note that based on our definition, $\Phi_{k_2}(x)$ can correct $k_2$ errors without Bob having received $\Phi_{k_1}(x)$ already. Hence, a question in the above process is that, after sending $\Phi_{k_1}(x)$, is $\Phi_{k_2}(x)$ necessary for Bob to recover $x$ if $d_E(x,y)\leq k_2$? Intuitively speaking, $\Phi_{k_1}(x)$ already gives Bob some information about $x$. So it is likely that Alice can send fewer bits than $\abs{\Phi_{k_2}(x)}$ for Bob to recover $x$ if $d_E(x,y)\leq k_2$.   The next theorem answers this question affirmatively. This is achievable because given $\Phi_{k_1}(x)$, we can construct a graph with smaller maximum degree and  apply Linial's algorithm to this graph.
\begin{restatable}{theorem}{thmincremental}\label{thm:incremental}
Let $a$ and $b$ be two positive integers such that $b>a$. 
Suppose $\Phi_a$ is a syndrome that can correct $a$ edits. If $y$ is obtained from $x$ by at most $b$ edits, then given $\Phi_a(x)$ one can construct a syndrome $\Phi_{b|a}:\Sigma_2^n \rightarrow [2^{(4b-2a)\log n+O(\log \log n)}]$ such that $x$ can be recovered given $y$, $\Phi_a(x)$, and $\Phi_{b|a}(x)$. 
\end{restatable}
Note that in the above theorem $\Phi_{b|a}$ has length approximately $(4b-2a)\log n$ which is less than $4b\log n$, the length of $\Phi_{b}(x)$. 

\paragraph*{Substring edit-correcting code} Our fourth result is  on codes correcting multiple substring edits. An $l$-substring edit in a string $x$ is the operation of replacing a substring $u$ of $x$ with another string $v$, where $|u|,|v|\le l$. An analysis of real and simulated data in~\cite{tang2023} reveals that substring edits are common in file synchronization and DNA data storage, where viewing errors as substring edits leads to smaller redundancies. We consider codes correcting $k$ $l$-substring edits, where $k$ is a constant, and have the following construction of codes with redundancy approximately twice the Gilbert-Varshamov bound. The method we use is to construct a graph with low maximum degree and apply Linial's algorithm.

\begin{theorem}\label{thm:SSE}
    Suppose $k$ is a positive integer, and $l$ is $O(\log n)$ or $\omega(\log n)$.  One can construct  codes of length $n$ correcting $k$ $l$-substring edits with redundancy $4k\log n+4kl +8k\log (l+1)+O(\log\log n)$. Moreover, the codes have polynomial-time encoders and decoders.
\end{theorem}

\subsection{Related Work}
\subsubsection*{Syndrome Compression}
In this work, we present a general method for constructing error-correcting codes based on distributed graph coloring. The codes have redundancy twice the Gilbert-Varshamov bound (\added{Levenshtein bound in the  edit distance case}) for combinatorial channels satisfying mild conditions. The only other general method we are aware of is  syndrome compression, introduced in a seminal paper by Sima et al.~\cite{sima2020optimal}. Starting from a reasonably good code, syndrome compression yields a code with redundancy twice the GV bound as well. We will discuss syndrome compression next, showing that it can be viewed as a specific method for recoloring a graph. We will also show that distributed graph coloring codes proposed here are simpler, more generalizable, and do not require us to start with a reasonably good code.

Syndrome Compression (SC) starts with a syndrome function $f$ that assigns distinct values to inputs that may produce the same output. In our framework, $f$ can be viewed as a proper coloring. Syndrome compression then produces a shorter syndrome using $f$. We can view this as a recoloring step, using number-theoretic arguments. Specifically, SC produces from a coloring with $2^k$ colors, a new one with $\Delta^22^{3.2k/\log (k/\log e)}$ colors. We may compare this to Linial's recoloring which starting from a $2^k$-coloring produces a new coloring with fewer colors, namely $5k\Delta^2$. For example, starting with $n$ bits to represent the colors ($2^n$ colors), after one round, we will need $2\log \Delta+O(n/\log n)$ and $2\log \Delta+\log n+O(1)$ bits to represent the colors (syndrome), with Syndrome Compression and Linial's recoloring, respectively. Because of this, the syndrome compression algorithm requires a pre-existing proper coloring $f$ with $2^{o(\log n\log\log n)}$ colors. Then, with recoloring, we will have $\Delta^2 2^{o(\log n)}$, or equivalently redundancy $2\log \Delta+o(\log n)$. A typical scenario is $k=O(\log n)$~\cite{sima2020syndrome}, which after recoloring leads to a redundancy of $2\log \Delta_n+O(\log n/\log\log n)$.

We can thus view syndrome compression as a specific distributed graph coloring (DGC) algorithm. However, utilizing Linial's algorithm for recoloring eliminates the need to have a reasonably good code to begin with. This allows us to construct codes for certain channels for a wider range of parameters with lower redundancy compared to the state of the art. Additionally, the DGC framework is conceptually simpler and more general, allowing us to formulate other problems such as list decoding with ease. Finally, as we show in this paper, DGC codes lead to a redundancy of $2\log \Delta_n+O(\log\log n)$, which converges to $2\log \Delta_n$ much faster than syndrome compression with redundancy $2\log \Delta_n+O(\log n/\log\log n)$.
\color{black}

\subsubsection*{Insertion/Deletion-Correcting Codes}
Most of our application is on codes correcting insertions and deletions. Codes correcting insertions and deletions were first introduced by Levenshtein in \cite{levenshtein1966binary}. Levenshtein shows that the Varshamov-Tenengolts
(VT) code can correct one deletion with redundancy $\log (n+1)$, where $n$ is the code length. He also shows that the optimal redundancy of codes correcting $k$ (a constant) deletions is between $k\log n+o(\log n)$ and $2k\log n+o(\log n)$. \added{For two deletions, Brakensiek et al.\cite{brakensiek2017efficient} was the first to achieve redundancy $O(\log n)$}. Gabrys et al.\ \cite{gabrys2018codes} propose codes correcting two deletions with redundancy $8\log n +O(\log \log n)$. Sima et al.~\cite{sima2019two} improve the redundancy to $7\log n +o(\log n)$. In \cite{guruswami2021explicit} Guruswami et al.\  further improve the redundancy to $4\log n+O(\log \log n)$, which matches the existential bound.

For $k$-deletion-correcting codes, where $k$ is a constant, Brakensiek et al.~ \cite{brakensiek2017efficient} construct codes correcting $k$ deletions with redundancy $O(k^2\log k\log n)$ and $O_k(n\log^4 n)$ encoding/decoding complexity. In \cite{sima2020optimal}, Sima et al.\ improve the redundancy to $8k \log n+o(\log n)$ but with encoding/decoding complexity $O(n^{2k+1})$. In \cite{sima2020syndrome}, Sima et al.\ propose syndrome compression  and further improve the redundancy to 
$4k(1+\epsilon)\log n$, where $\epsilon=O(1/k)+O(k^2\log k/\log\log n)$,   with encoding/decoding complexity $O(n^{2k+1})$. Sima et al.\ also use syndrome compression to construct systematic codes correcting $k$ deletions with redundancy $4k\log n+o(\log n)$ \cite{sima2020optimalt}. Our Corollary~\ref{cor:kedit} achieves the state-of-the-art of $k$-edit-correcting codes, which construct binary $k$-edit-correcting codes with redundancy $4k\log n+O(\log \log n)$ and polynomial time encoding and decoding complexity.  
\added{As for $k$-deletion-correcting codes over larger alphabet size, in \cite{sima2020optimalq}, for $q\leq n$, Sima et al. construct codes correcting $k$ deletions with redundancy $4k\log n+2k\log q+ o(\log n)$. For $q>n$, Sima et al. \cite{sima2020optimalq} construct codes correcting $k$ deletions with redundancy $(30k+1)\log q$. In \cite{haeupler2017synchronization}, Haeupler et al.\ propose a  technique called `synchronization strings' and  construct codes correcting edits with rate approaching the Singleton bound over large constant alphabet. In \cite{cheng2018deterministic} and \cite{haeupler2019optimal}, Cheng et al.\ and Haeupler respectively construct binary codes correcting $\epsilon$ fraction of edits with rate $1-O(\epsilon\log^2\frac{1}{\epsilon})$. }

\subsubsection*{List-Decodable Deletion Codes}
For list-decodable codes correcting a constant number of deletions, in \cite{Wachter2018list}, Wachter-Zeh shows that VT code can correct $k$ deletions with list size at most $n^{k-1}$. In \cite{guruswami2021explicit}, Guruswami et al. construct codes correcting 2 deletions with list size 2 and redundancy $3\log n$.

\subsubsection*{A Burst of Deletions} Codes correcting a burst of deletions have been intensively studied.
In 
\cite{Schoeny17}, Schoeny et al.\ 
construct codes correcting a burst of exact $l$ deletions with redundancy $\log n+(l-1)\log\log n+l-\log l$. In 
\cite{lenz20}, Lenz et al.\  construct codes correcting a burst of deletions
of variable length with redundancy $\log n+(l(l+1)/2)\log\log n+c_l$, where 
$c_l$ is a constant that only depends on $l$.  
In 
\cite{bitar21}, Bitar et al.\    construct codes correcting  deletions within a window of length $l$ with redundancy $\log n+O(l\log^2 (l \log n))$. In fact, a simple counting argument shows that optimal codes correcting a burst of at most $l$ deletions have redundancy at least $l+\log n$. So the above codes are all asymptotically optimal when $l$ is a constant. However, when $l$ is large, for example, $l=\omega( \log n)$,  the above codes have large redundancy. In \cite{tang2023} and \cite{yuting2024asymptotically}, Tang et al.\ and Li et al.\ construct codes correcting an $l$-substring edit with redundancy $2\log n+o_l( \log n)$ and $\log n+O_l(\log \log n)$ respectively. However, their method does not work when $l=\omega( \log n)$. In contrast, we will construct asymptotically optimal codes correcting a burst of at most $l$ deletions for $l=\omega(\log n)$.

 \subsection{Basic Definitions and Notations}\label{subsec:not}
Let $\Sigma_q$ be the alphabet $\{0,1,\dotsc,q-1\}$ and $\Sigma_q^n$ denote all strings of length $n$ over $\Sigma_q$.
For a string $x$, $\abs{x}$ denotes its length.
For a set $S$, $\abs{S}$ denotes the cardinality of $S$. Logarithms in the paper are to the base 2. The redundancy of a code $\cC \subset \Sigma_2^n$ is defined as $ n-\log \abs{\cC}$ bits. Suppose $G$ is a graph and $v$ is a vertex. We use $\nei(v)$ to denote the neighborhood of $v$ and use $\deg(v)$ to denote the degree of $v$. We use $d_E$ to denote edit distances. For a string $x$ and an integer $a$, let  $B_a(x):=\{y\in \Sigma_2^*:d_E(x,y)\leq a\}$.
 
\begin{definition}[Channels]
    We characterize a combinatorial communication/storage \emph{channel} by an operator $\cE$, where for $x\in\Sigma_q^n$, $\cE(x)$ denotes the set of all possible channel outputs for the input $x$. Furthermore, for $y\in \Sigma_q^*$, 
 \begin{align*}
 \cE^{-1}(y):=\left\{x\in \Sigma_q^n: y\in \cE(x) \right\},
 \end{align*}
 which consists of all the length $n$ strings that can produce $y$ using the errors in $\cE$.
\end{definition} 
\begin{example}
     Suppose $n=3$, $q=2$ and $\cE$ is one deletion. Let $x=010$, $y=01$. Then $\cE(x)=\{01,00,10\}$, $\cE^{-1}(y)=\{001,101,011,101,010\}$.
\end{example}

\begin{definition}[Confusion Graph and Hypergraph]
    The \emph{confusion graph $G_n$ for the channel $\cE$ with inputs $\Sigma_q^n$} is the graph whose vertex set is $\Sigma_q^n$ and $(x,x')$ is an edge in $G_n$ if there exists $y\in\Sigma_q^*$ such that $x,x'\in\cE^{-1}(y)$. 

    The \emph{confusion hypergraph $H_n$ for the channel $\cE$ with inputs $\Sigma_q^n$} is a hypergraph whose vertex set is $\Sigma_q^n$ and edge set is $\{\cE^{-1}(y):y\in\Sigma_q^*\}$.
\end{definition}
Examples are given in \Cref{fig:confGraph,fig:confHyperGraph}.

For a graph $G$, we use $V(G)$ to denote its vertex set. A graph $G$ is $T$-\textit{explicit} if there is an ordering of vertices and an algorithm such that given a vertex index $v$ and $1\le i\le \deg v$, the algorithm can output $v$'s $i$-th neighbor in this ordering in time $T$. 

For a hypergraph $H$, we use $G(H)$ to denote the graph associated with $H$. Concretely, $G(H)$ and $H$ have the same vertex set.  Suppose $u$ and $v$ are vertices of $H$,  then $(u,v)$ is an edge in $G(H)$ if and only if $u$ and $v$ are in the same edge of $H$. For a hypergraph $H$, we use $r(H)$ to denote the maximum degree of $H$ and use $v(H)$ to denote the size of the maximum edge. For a hypergraph $H$, we use $V(H)$ to denote its vertex set. 

We next define systematic error-correcting codes for a given channel. Note that this definition slightly deviates from the conventional definition as discussed in \Cref{rem:err-free}.
\begin{definition}[c.f. \Cref{rem:err-free}]
    The set
    \[
    \cC = \{(x,\Phi(x)):x\in \Sigma_q^n\},
    \]
    is a \emph{systematic error-correcting code} for a \added{synchronization} channel $\cE$  if for any $x,x'\in \Sigma_q^n,y\in\Sigma_q^*,x,x'\in\cE^{-1}(y)$, we have $\Phi(x)\neq\Phi(x')$. The redundancy of this code is $\max_{x\in\Sigma_q^n}\abs{\Phi(x)}$.
\end{definition}

As discussed before, in the context of LOCAL algorithms, each processor (vertex) is infinitely powerful as the model focuses on the locality of information. In our problem however, we are also interested in computational complexity. Additionally, the fact that LOCAL algorithms are defined for distributed computing systems may be confusing. To avoid ambiguity, we provide the following definition.
\begin{definition}[Time complexity of algorithms]
    Suppose $G$ is a graph (hypergraph), $A$ is a set, and $f: V(G) \rightarrow A$ is a function. We say that $f$ is \emph{individually  computable in time $T$} if it can be determined for a \textit{given} vertex $v$ by performing $T$ computational operations.
\end{definition}

% In some cases, it may be the case that elements in the set $\sigma_{\cE}$ belong to a different space than $\Sigma_q^n$. For instance, for the case where $\cE$ consists of exactly $k$ deletions, the set $\sigma_{\cE}$ contains elements from $\Sigma_q^{n-k}$. In this case, the graph $G_{n,\cE}$ is a bi-partite graph $G_{n,\cE}=\left( (V_1, V_2), E \right)$ where $V_1 = \Sigma^n_{q}$ and $V_2=\Sigma^{n-k}_{q}$, and the edge set is defined as before using the operators from $\cE$.

 As stated in the introduction, we will be focused on the regime where  $\cE$ is such that $\abs{\cE(x)} = O(n^a)$ and  $\abs{\cE^{-1}(y)}=O(n^b)$ for some constants $a$ and $b$ and all $x \in \Sigma_q^n,y\in\Sigma_q^*$.

\section{Construction of Codes Correcting a Constant Number of Errors}\label{sec:wcerrors}
In this section, we first propose a general framework of using distributed graph coloring algorithms to construct error-correcting codes. Then we use a specific algorithm --- Linial's algorithm to construct colorings of the confusion graph and construct  codes that leverage cover-free set families capable of recovering from a constant number of errors. Finally, we compare our method with syndrome compression and show that syndrome compression can be viewed as a recoloring process.

\subsection{From Distributed Graph Coloring to Code Construction}
In distributed graph algorithms, each node in an underlying graph $G$ acts as a processor that can only communicate with its immediate neighbors in synchronous rounds. In the LOCAL model, which is of interest in this paper, in each round, every node can send a message to each of its neighbors, receive all messages sent to it, and perform an unbounded amount of local computation before moving on to the next round.

The goal is typically to solve global graph-theoretic problems, such as coloring and finding maximal independent sets, by transforming local data into global solutions. For a given problem, it is of interest to determine the number of rounds needed to arrive at the solution. 

A result of the assumptions made in the LOCAL model is that, for an algorithm that takes $R$ rounds, the part of the solution computed at each node depends only on its neighborhood of radius $R$. This fact is used in the next theorem, which relates LOCAL graph coloring algorithms to the construction of codes with efficient encoding and decoding for any channel.

\thmrt*
\begin{proof}
For any vertex $x$, the $R_n$-th round result $\Phi(x)$ is determined by the computation result of vertices whose distance from $x$ is at most $R_n$. (In fact, from a node at distance $R'<R_n$ from $x$, we only need the first $R_n-R'+1$ messages.)  So the amount of computation needed to compute $\Phi(x)$ is bounded by 
$$(1+\Delta_n+\Delta_n^2+\dotsc +\Delta_n^{R_n})R_nT_n,$$ which is at most $\Delta_n^{R_n}(R_n+1)R_nT_n$.
% If $R_n$ is a constant and $T_n$, $\Delta_n$ are polynomials in $n$, then to compute $\Phi(x)$ takes polynomial time for any $x$. 

\added{Moreover, to decode $y$, one needs to compute $\Phi(x)$ for all $x\in \cE^{-1}(y)$. Since $\abs{\cE^{-1}(y)}$ is bounded by $\Delta_n$, the decoding process takes $O(\Delta_n^{R_n+1}(R_n+1)R_nT_n)$. }
\end{proof}
This theorem enables us to benefit from the research on distributed graph coloring under the LOCAL model \added{\protect\cite{kuhn09}\protect\cite{maus2023distributedgraphcoloringeasy}\protect\cite{barenboim2013distributed}}. In particular, if $B_n=O(\Delta_n^2$), then the redundancy is only twice the GV bound. However, there are two challenges that need to be addressed to achieve polynomial-time encoding and decoding with this redundancy. First, the number of rounds must be constant, and second, $\Delta_n,T_n$ must be polynomial in $n$. Unfortunately, existing algorithms with $O(\Delta_n^2)$ colors run in $\log^* n$ rounds, with no limit on $T_n$. These are addressed in the next subsection.

\subsection{Cover-Free Set Families}
 The constructions of our colorings are based on combinatorial objects called cover-free set families. Informally, a cover-free set family is a family of sets in which no set is covered by a limited number of other sets. The formal definition is given next.
\begin{definition}\label{def:cff}
A family $\cJ=\{F_1,F_2,\dotsc,F_N\}$ of subsets  of a ``ground'' set $M$ is called \emph{$r$-cover-free} over $M$ if for any indices $i_0,i_1,\dotsc,i_{r}\in [N]$, where $i_0\not\in \{i_1,\dotsc,i_{r}\}$, $F_{i_0}\setminus\bigcup_{j=1}^{r} F_{i_j}\neq \emptyset$. 
\end{definition}

\added{In the next subsection, we will use cover-free families in a graph recoloring scheme for code construction as follows: Suppose we have a graph $G$ with maximum degree $\Delta$ and an old coloring of size $N$. A $\Delta$-cover-free family of size $N$ can be used to recolor the graph. We associate each old color with a set in the $\Delta$-cover-free family. The new color of a vertex is an element of the set $F\setminus \cup_{i=1}^r F_i$, where $F$ is the set associated with the vertex and $F_i$ is the set associated with the $i$th neighbor of that vertex. Below, we will detail how these families are constructed for our application.}

In the LOCAL model, the goal is usually to minimize the number of rounds of interaction. However, if we want to apply local graph coloring algorithms to generate efficient error-correcting codes, then we need to also account for the time complexity of the computation, which motivates our next definition. 
\begin{definition}\label{def:wit}
Suppose $\cJ=\{F_1,F_2,\dotsc,F_N\}$  is an \emph{$r$-cover-free} family over $M$. We say $\cJ$ can be witnessed in time $T$ if there is a function $W$, which we refer to as a witness algorithm, such that for any input indices $i_0,i_1,\dotsc,i_{r}\in [N]$, where $i_0\not\in \{i_1,\dotsc,i_{r}\}$, $W(i_0,\{i_1,\dotsc,i_{r}\})$ outputs a witness element of $F_{i_0}\setminus\bigcup_{j=1}^{r} F_{i_j}$ in time $T$.
\end{definition}

 The first part of the next lemma is a classical construction of cover-free set families given by Erdős et al.\ \cite{Erds1985FamiliesOF}. For our purposes, we will analyze the construction in terms of potential witness algorithms. 

\begin{lemma}[c.f.~\cite{Erds1985FamiliesOF}]\label{lem:constructcff}
    Let $Q$ be a prime, $b,r$ be non-negative integers, and $M=\mathbb{F}_Q\times  \mathbb{F}_Q.$ Let $\cJ=\{F_1,F_2,\dotsc,F_{Q^{b+1}}\}$, where 
\begin{align*}\label{eq:constFq}
    F_i=\{(\alpha,g_i(\alpha)):\alpha\in \mathbb{F}_Q\},
\end{align*}
where $g_i$ is the $i$-th polynomial in $\mathbb{F}_Q[x]$ in the lexicographic ordering with degree $\leq b$.
 If $br<Q$, then $\cJ$ is an $r$-cover-free set family. Moreover, $\cJ$ has a witness algorithm that runs in time $\poly(Q)$.\label{lem:constructcff1}
\end{lemma}

\begin{proof}
  It suffices to construct a witness algorithm $W$. Given inputs $i_0,i_1,\dotsc,i_{r}\in [Q^{b+1}]$, $W$ computes 
    \begin{align*}
        f(x)=\prod_{j=1}^{r}(g_{i_0}(x)-g_{i_j}(x)).
    \end{align*}
    Because $0<\deg f\leq br<Q$, $f$ has at least one non-zero point. $W$ searches for a non-zero point $\alpha$ and outputs $(\alpha,g_{i_0}(\alpha))$. It is clear that  
    \begin{align*}
        (\alpha,g_{i_0}(\alpha))\in F_{i_0}\setminus\bigcup_{j=1}^{r} F_{i_j}.
    \end{align*}
    Thus $W$ is a witness algorithm.    Moreover, $W$ runs in time $\poly(Q)$.
\end{proof}

Taking suitable parameters gives us the following corollary. 
\begin{corollary}\label{cor:constructw}
    For any $r$, there exists an $r$ cover-free set family $\cJ=\{F_1,F_2,\dotsc,F_N\}$  over $[t]$, where $t=O(r^2\log ^2N)$. $\cJ$ has a witness algorithm that runs in time $\poly(r,\log N)$.    
\end{corollary}
\begin{proof}
    Apply Lemma~\ref{lem:constructcff} with $b=\ceil{\log N}$ and $Q$ being the smallest prime power that is greater than $rb$.
\end{proof}

In the next section we apply these results to the problem of error-correcting codes.

\subsection{Codes Correcting a Constant Number of Errors}
\color{black}
In this subsection, we will construct a LOCAL graph coloring algorithm for $G_n$ based on Linial's algorithm that runs in two rounds where $T_n$ is a polynomial, thus yielding efficiently encodable and decodable codes for a channel $\cE$. We focus mainly on the case where $\cE$ represents a constant number of errors. %In this case, $\Delta(G_n)=\poly(n)$. 
In this section, we consider graphs with $2^{\poly(n)}$ vertices and maximum degree $\poly(n)$. We show that for this class of graphs, two rounds are sufficient to get a coloring of size $\Tilde{O}(\Delta(G)^2)$. 
% \begin{definition}\label{def:explicit}
%     A graph $G$ is $T$-explicit if there is an order and an algorithm $A$ such that given a vertex index $v$ and $i\in [d(v)]$, $A$ can output $v$'s $i$-th neighbor in this order in time $T$.
% \end{definition}

If $\Phi :V(G_n)\rightarrow M$ is a coloring that can be individually computed in time $T_{\Phi}$, then we get a coloring that can be used to correct the errors in $\cE$ whose complexity can be bounded as a function of $T_{\Phi}$. Suppose  $G_n$ is $T_e$-explicit. If $y\in\cE(x)$, then given $y$ and $\Phi(x)$, one way to recover $x$ is by computing  $\Phi(z)$ for all $z\in \cE^{-1}(y)$ and compare with $\Phi(x)$, which takes time  $\Delta(G_n)(T_e+T_{\Phi})$. Next we show how to use a $\Delta(G)$-cover-free subset family to get a new coloring of smaller size from an existing coloring.

\begin{lemma}\label{lem:cfftocolor}
    Suppose $G$ is a $T_e$-explicit graph. Suppose $\cJ=\{F_1,F_2,\dotsc,F_N\}$ is a $\Delta(G)$-cover-free subset family over a ground set $M$ that can be witnessed in time $T_w$. If we have an existing coloring $\Phi :V(G)\rightarrow [N]$ that can be individually computed in time $T_{\Phi}$, then we can get a new coloring $\Phi^{*} :V(G)\rightarrow M$, which can be individually computed in time $O(\Delta(G)(T_e+T_{\Phi})+T_w)$.
\end{lemma}

\begin{proof}
Let $W$ denote the witness algorithm of $\cJ$.
    For a vertex $v$, we first compute $\Phi(v)$ and $\Phi(\nei(v))$, which takes time $(\Delta(G)+1)(T_\Phi + T_e)$.
   % $(\Delta(G)+1)(T_e+T_\Phi)$. 
   Note that $\abs{\Phi(\nei(v))}\leq \Delta(G)$. Let $\Phi^{*}(v)$ be $W(\Phi(v),\Phi(\nei(v)))$. The whole process takes time $(\Delta(G)+1)(T_\Phi + T_e) + T_w$.
   
   %$O(\Delta(G)(T_e+T_h)+T_w)$.

    We now prove that $\Phi^{*}$ is a coloring. Assume $u\in \nei(v)$. We will show that $\Phi^{*}(v)\neq \Phi^{*}(u)$.    
    Since $\Phi$ is a coloring, we have $\Phi(v)\not \in \Phi(\nei(v)) $.  Since  $W$ is the witness algorithm defined over the family of sets $\cJ$, it holds that  $\Phi^{*}(v)=W(\Phi(v),\Phi(\nei(v)))\in F_{\Phi(v)}\setminus F_{\Phi(u)}$. Similarly we have $\Phi^{*}(u)\in F_{\Phi(u)}\setminus F_{\Phi(v)}$. Thus $\Phi^{*}(v)\neq \Phi^{*}(u)$.
\end{proof}

%Next we will prove that for a graph with $2^{\poly(n)}$ vertices and maximum degree $\poly(n)$ 
Next we prove several useful properties pertaining to coloring a graph with $2^{\poly(n)}$ vertices and maximum degree $\poly(n)$ that use Lemma~\ref{lem:cfftocolor}. The first statement in the next lemma follows from a single application of Lemma~\ref{lem:cfftocolor} whereas the second statement follows from two iterations of the same lemma.

%applying Lemma~\ref{lem:cfftocolor} twice suffices to obtain a coloring of size $\tilde{O}(\Delta(G)^2)$. 

\begin{lemma}\label{lem:deltaG2}
    Suppose $p(n)$ is upper bounded by a polynomial in $n$ and $G$ has $2^{p(n)}$ vertices. Suppose $\Delta (G)=\poly(n)$ and $G$ is $\poly(n)$-explicit. Then one can construct
    \begin{enumerate}
        \item a coloring $\Phi:V(G)\rightarrow [O(\Delta(G)^2 p(n)^2)]$, which can be individually computed in time $\poly(n)$. \label{lem:deltaG21}
        \item a coloring $\Phi:V(G)\rightarrow [\tilde{O}(\Delta(G)^2)]$, which can be individually computed in time $\poly(n)$.\label{lem:deltaG22}
    \end{enumerate}
    
\end{lemma}

\begin{proof}
\begin{enumerate}
    \item Take $N=2^{p(n)}$ and  $r=\Delta(G)$.  Corollary~\ref{cor:constructw} gives a $\Delta(G)$-cover-free family with $N$ sets over a ground set of size $O(p(n)^2\Delta(G)^2)$, which can be witnessed in time $\poly(n)$. Now, taking the identity of each vertex to be its old color, by Lemma~\ref{lem:cfftocolor}, one can construct a new coloring $\Phi:V(G)\rightarrow [O(\Delta(G)^2 p^2(n))]$, which can be individually computed in time $\poly(n)$.
    \item Let $N'=\poly(n)$. By 1), one can construct a coloring $\Phi:V(G)\rightarrow [N']$, which can be individually computed in time $\poly(n)$.
    With the assumed value of $N'$ and with $r=\Delta(G)$, Corollary~\ref{cor:constructw} gives a $\Delta(G)$-cover-free family with $N'$ sets over a ground set of size $O(\Delta^2(G)\log^2 N')$, which is $\tilde{O}(\Delta^2(G))$ since $N'=\poly(n)$. Moreover, it can be witnessed in time $\poly(n)$.  So by Lemma~\ref{lem:cfftocolor}, one can construct a new coloring $\Phi^{*}:V(G)\rightarrow [\tilde{O}(\Delta^2(G))]$, which can be individually computed in time $\poly(n)$.
\end{enumerate}
\end{proof}

We now restate Theorem~\ref{thm:GV}, which is an immediate consequence of the previous lemma. 

\thmGV*
\begin{proof}
    Lemma~\ref{lem:deltaG2} provides a LOCAL graph coloring algorithm for $G_n$ of size $2^{\rho}$ that runs in two rounds, and the time complexity of both rounds at each node is polynomial in $n$. Hence, by Theorem~\ref{thm:rt}, encoding and decoding takes polynomial time.
\end{proof}

Now we can apply Theorem~\ref{thm:GV} to $G_n$, where $\cE$  represents a constant number of edits (and where each edit can be an insertion or deletion). For example, if we take $\cE$ to be $k$ edits, where $k$ is a constant, then we can get the following corollary.

\begin{corollary}\label{cor:kedit}
 Suppose $k$ is a constant and $q(n)=O(\poly(n))$.    There is a function $\Phi:\Sigma_{q(n)}^n\rightarrow [2^{4k\log n+4k\log q(n)+O(\log\log n)}]$ that can correct $k$ edits. Specifically, if $y$ is obtained by $k$ edits from $x$, then given $y$ and $\Phi(x)$ one can recover $x$. Moreover, computing $\Phi$ and the recovery process takes $\poly(n)$ time.
\end{corollary}
\begin{proof}
    Let $\cE$ be $k$ edits, then $G_{n}$ has $q(n)^n=2^{O(\poly( n))}$ vertices.  
$\Delta(G_{n})=O(n^{2k}q(n)^{2k})=O(\poly(n))$. \added{Moreover, $G_{n}$ is $\poly(n)$-explicit. This is because assuming we are looking for the $j$th neighbor of $x\in\Sigma_2^n$. It takes $O(1)$ time to get the $(i+1)$th neighbor given the $i$th neighbor. In addition, the number of neighbors is bounded by $\Delta(G_n)=\poly(n)$. Thus, it takes $\poly(n)$ time to output the $j$th neighbor. So, $G_{n}$ is $\poly(n)$-explicit. }So we apply Theorem~\ref{thm:GV} and complete the proof.
\end{proof}

\subsection{Syndrome Compression as Recoloring with Cover-free Set Families}\label{subsec:sc}
We now show that syndrome compression \cite{sima2020syndrome} can be viewed as a specific method for constructing cover-free set families, and then using them to recolor to obtain, from a coloring of size $N=2^{o(\log n\log\log n)}$, a new coloring of size $O\left(\Delta^22^{o(\log n)}\right)$. In contrast to the typical use of polynomials to construct cover-free set families for Linial's graph coloring, syndrome compression uses an upper bound on the number of divisors, as described next. 

\begin{fact}[\hspace{1sp}{\cite[Lemma 3]{sima2020syndrome}, cf.  \cite{MR938967}}]\label{fact:divisor}
For a positive integer $N\geq 3$, the number
of divisors of $N$ is upper bounded by 
\[ 2^{1.6\frac{\log N}{\log (\log N/\log e)}}.\]
\end{fact}

\begin{construction}\label{syndromeCompression}
Let $N$ and $r$ be positive integers and $$A=r2^{1.6\frac{\log N}{\log (\log N/\log e)}}+1.$$ 
Define $$F_i=\{(a,i \bmod a):a\in [A]\}, i\in [N].$$ Construct $\cJ=\{F_1,\dotsc,F_N\}$.
\end{construction}

\begin{lemma}
    The family $\cJ$ created in Construction~\ref{syndromeCompression} is an $r$-cover-free set family over $[A]\times [A]$. The size of $\cJ$ is $N$, and the size of the ground set is $$O\left(r^22^{3.2\frac{\log N}{\log (\log N/\log e)}}\right).$$
\end{lemma}

\begin{proof}
For each $F_i, F_j$ ($i\neq j$), $\abs{F_i\cap F_j}$ is the number of divisors of $\abs{i-j}$,  which is at most $$2^{1.6\frac{\log N}{\log (\log N/\log e)}},$$ by Fact \ref{fact:divisor}. 
So,  $r\abs{F_i\cap F_j}<A$ for $i\neq j$. Hence, $$F_{i_0}\setminus\bigcup_{j=1}^{r}F_{i_j}\neq \emptyset$$  if $i_0 \not\in \{i_1,\dotsc,i_r\}$. Therefore, $\cJ$ is an $r$-cover-free set family.
\end{proof}

Theorem~\ref{syndromeCompressionTheorem} is a restatement of syndrome compression obtained by applying Linial's algorithm (\Cref{lem:cfftocolor}) with respect to the cover-free set families given in Construction~\ref{syndromeCompression}.

\begin{theorem}\label{syndromeCompressionTheorem}
 Suppose there exists an existing coloring of $G$ of size $2^{o(\log n \log \log n)}$. Then one can get a new coloring of size $O\left({\Delta(G)}^2 2^{o(\log n)}\right)$.
\end{theorem}

\begin{proof}
By assumption, $G$ has an existing coloring of size $N$, where $N=2^{w(n)\log n}$ and $1<w(n)=o(\log \log n)$.
By Construction \ref{syndromeCompression}, there exists a $\Delta$-cover-free set family of size $N$ over a ground set of size $$O\left(\Delta^22^{3.2\frac{w(n)\log n}{\log(w(n)\log n/\log e)}}\right)=O\left(\Delta^2 2^{o(\log n)}\right).$$ By \Cref{lem:cfftocolor}, we can get a new coloring of size $O\left(\Delta^2 2^{o(\log n)}\right)$.
\end{proof}

The main difference between our method and syndrome compression is the construction of $r$-cover-free set families. Note that to get an $r$-cover-free set family of size $k$, Construction \ref{syndromeCompression} requires a larger ground set
compared to \Cref{cor:constructw}. In other words, for a specific ground set, Construction \ref{syndromeCompression} allows a smaller $r$-cover-free set family than \Cref{cor:constructw}. Therefore, if one wants to apply Construction \ref{syndromeCompression} and \Cref{lem:cfftocolor} to get a new coloring, one needs an existing coloring of a small size in advance.  This restriction may complicate potential code constructions. By contrast, in our method, we do not need an existing coloring of size $2^{o(\log n\log \log n)}$, and so the code construction provides more flexibility. 

Moreover, our method can be used for parameters that one cannot use syndrome compression. For example, in \cite{sima2020optimalq}, Sima et al.\ study $q$-ary deletion codes. They show that syndrome compression can only be applied if the alphabet size $q(n)$ is smaller than $n$. In contrast, Corollary~\ref{cor:kedit} can deal with the case where $q(n)=\poly(n)$.
Furthermore, in \cite{sima2020optimalq}, Sima et al.\  construct codes correcting $k$ deletions in $q(n)$-ary sequences for $q(n)>n$ with redundancy $(30k+1)\log q(n)$. By contrast, for  $q(n)=\poly(n)$, our redundancy in Corollary~\ref{cor:kedit} is $4k\log n+4k\log q(n)+O(\log\log n)$, which is much smaller.

\section{List-Decodable Codes Correcting a Constant Number of Errors}
In this section, we consider construction of list-decodable codes using methods generalized from local graph coloring.
The notion of a defective coloring, which will be introduced next, will be useful.

At first glance, it appears that a  defective coloring of $G_n$ will be of use. 
\begin{definition}
    Suppose that $G$ is a graph and $A$ is a set. A function $\Phi:V(G)\to A$ is a $d$-defective coloring if for any vertex $u$, there are at most $d$ neighbors $v\in\nei(u)$ satisfying $\Phi(u)=\Phi(v)$. 
\end{definition}
Note that a coloring is just a $0$-defective coloring. If we have an $(\ell-1)$-defective coloring $\Phi:G_n\to A$, then if the sender sends $x\in\Sigma_q^n$ and the receiver receives $y\in\Sigma_q^*$, then $\cE^{-1}(y)\subset \{x\}\cup \nei(x) $. So there are at most $\ell$  sequences $x'\in \cE^{-1}(y)$ satisfying that $\Phi(x)=\Phi(x')$. So one can list decode with list size $l$ using  $\Phi(x)$. If one wants to utilize local graph algorithms for defective
coloring to construct list-decodable codes, one will not get useful results. For example, in \cite{kuhn09}, Kuhn proposes a local  algorithm outputting an $O(\Delta^2/d^2)$-coloring with defect $d$, where $\Delta$ is the maximum degree of the graph. If one uses this result with $G=G_n$ and $d=\ell-1$, where $\ell$ is the list size, then one can get an $\ell$-list-decodable code with redundancy $2\log \Delta(G_n)-2\log (\ell-1)$, which is approximately twice the Gilbert-Varshamov bound (\added{Levenshtein bound in the  edit distance case}) if the list size is a constant. So we do not get much improvement compared to unique decoding. Moreover, polynomial-size lists are not helpful since for any $y$, $\abs{\cE^{-1}(y)}$ is already 
bounded by a polynomial in $n$.

In fact, the notion of a defective coloring is too strong for our purpose. 
Suppose that the sender sends $x\in \Sigma_q^n$ and the receiver receives $y\in\Sigma_q^*$, we do not need to require that $\Phi:V(G_n)\rightarrow A$ is an $(\ell-1)$-defective coloring. What we really need is  that $\Phi(z)=\Phi(x)$ holds for at most $\ell$ sequences 
 $z\in \cE^{-1}(y)$. In other words, for any  $x$ and any $y\in\cE(x)$, we have 
 \begin{align}\label{con:star}
     \abs{\{z\in\cE^{-1}(y):\Phi(z)=\Phi(x)\}}\leq \ell.
 \end{align}  When $\ell=1$, i.e. the unique decoding case,
the two conditions are the same. But when $\ell>1$, \eqref{con:star} is looser than defective coloring.
We now define the notion of $\ell$-labelings, which  captures the precise requirement of constructing list-decodable codes with list size $\ell$.

\begin{definition}\label{def:lcoloring}
Let $H$ be a hypergraph, and $\Phi:V(H)\rightarrow A$ is a function. We say that $\Phi$ is an $\ell$-labeling if for any $a\in A$ and edge $E$ of $H$, the equation $\Phi(x)=a$ has at most $\ell$ solutions in $E$.

An algorithm $\Dec$ is called an $\ell$-decoder of $\Phi$ if it takes an element $a$ of $A$, and an edge $E$ of $H$ as its inputs, and its output $\Dec(a,E)$ is a set of vertices of size at most $\ell$ satisfying that: 
for any  edge $E$ of $H$ and any  vertex $u\in E$, it holds that $u\in \Dec(\Phi(u),E)$.
\end{definition}

Note that $\Phi$ is an $\ell$-labeling if and only if $\Phi$ has an $\ell$-decoder.
If we have an $\ell$-labeling of $H_n$, which is $\Phi$, then $\{(x,\Phi(x)):x\in\Sigma_q^n\}$ is a list-decodable code for $\cE$ with list size $\ell$  (the $\ell$-decoder is the decoder of the code). Conversely, if we have a systematic list-decodable encoder $E(x)=(x,\Phi(x))$ with list size $\ell$, then $\Phi$ is an $\ell$-labeling.

\begin{example}
    Let $H_3$ be the hypergraph in Figure~\ref{fig:confHyperGraph}. Then $\Phi(000)=\Phi(010)=\Phi(101)=\Phi(111)=\text{red (circle)}$ and $\Phi(100)=\Phi(001)=\Phi(110)=\Phi(011)=\text{blue (rectangle)}$ is a $2$-labeling of $H_3$. For example, we have $\Dec(\text{red},00)=\{000,010\}$, $\Dec(\text{red},01)=\{010,101\}$, $\Dec(\text{blue},11)=\{110,011\}$, and $\Dec(\text{blue},10)=\{100,110\}$.
\end{example}

\subsection{$(r,v,l)$ Cover-free Set Families}
We will use a generalization of cover-free set families to construct labelings. 
\begin{definition}\label{def:rvlfree}
    A set family $\cJ=\{F_1,F_2,\dotsc,F_N\}$ over a ground set $M$ is called $(r,v,\ell)$-cover-free set family if for any $u\in [N]$, $U_1,U_2,\dotsc,U_r\subset [N]$ where $u\in U_i$ and $\abs{U_i}\leq v$ for all $i\in[r]$, there exists $x\in F_u$ such that the following holds: For all $i\in [r]$, $x\in F_s$ holds for at most $\ell$ values of $s\in U_i$. 
\end{definition}
In other words, for any set $F_u$ in $\cJ$, and any $r$ groups of sets that each include $F_u$ and contain at most $v$ sets, there is at least one element in $F_u$ that appears in no more than $\ell$ sets within each group. Note that Definition~\ref{def:rvlfree} is a generalization of cover-free set families. $\cJ$ is an $r$-cover-free set family if and only if $\cJ$ is an $(r,2,1)$-cover-free set family. 

Next, we prove the existence of $(r,v,\ell)$-cover-free set family.
\begin{lemma}\label{lem:rvlexist}
    For any positive integers $r,v,\ell,N$, there exists an  $(r,v,\ell)$-cover-free set family of size $N$ over a ground set  $[t]$, where $t=O(r^{1+1/\ell}v^2\log N)$.
\end{lemma}
\begin{proof}
Let $t=\ceil{6 r^{1+1/\ell}v^2 \log N}$ and  $p=\left(\frac{1}{(\ell+1)r}\right)^{1/\ell}\frac{1}{v}$. We construct each $F_u,u\in[N],$ by including each element of $[t]$ in it with probability $p$, independent of all other choices. %choose $F_1$ randomly and let $p$ be the probability of each element of $[t]$ being in $F_1$. We choose $F_2,\dotsc,F_N$ in the same way independently.   
We show $\cJ:=\{F_1,F_2,\dotsc,F_N\}$ is an $(r,v,\ell)$-cover-free set family with positive probability.

We call $(u, U_1,\cdots ,U_r)$ an obstruction of $\cJ$ if 
\begin{enumerate}
    \item $u\in[N], U_1,\cdots ,U_r\subset [N]$ \label{obs: 1}
    \item $u\in U_i$ and $\abs{U_i}= v$ for  $i\in [r]$ \label{obs: 2}
    \item for all $x\in [t]$ \label{obs: 3}, either $x\not\in F_u$, or $x\in F_u$ and  $\exists i \in [r]$,  distinct $s_1,s_2,\dotsc,s_{\ell}\in U_i\setminus\{u\}$,  such that $x\in F_{s_j}$  for $j\in[\ell]$.
\end{enumerate}
Then $\cJ$ is an $(r,v,\ell)$-cover-free set family if and only if  $\cJ$ has no obstruction. For any $(u, U_1,\cdots ,U_r)$ that satisfies \ref{obs: 1}) and \ref{obs: 2}) the probability that it satisfies \ref{obs: 3}) is at most $\left( 1-p+ r v^{\ell} p^{\ell+1}\right)^t$. Thus, letting $t=\ceil{6 r^{1+1/{\ell}}v^2 \log N}$, we have
\begin{align*}
    \Pr&[\cJ \text{ is not an }(r,v,\ell)\text{-cover-free family}]\\
    &=\Pr[\cJ \text{ has at least one obstruction}] \\
    & \leq N^{rv}\left( 1-p+ r v^{\ell} p^{\ell+1}\right)^t \\
    & \leq N^{rv}\left( 1-\frac{\ell}{\ell+1}p\right)^t\\
    &<N^{rv}e^{-\frac{\ell}{\ell+1}pt}\\
     &<N^{rv}e^{-2rv\ln N}\\
     &<N^{-rv}\\
    &<1.
\end{align*}
Hence, there exists a choice of $\cJ$ that is an $(r,v,\ell)$-cover-free set family.
\end{proof}
Note that the proof of Lemma~\ref{lem:rvlexist} shows a random set family constructed as described in the lemma is in fact $(r,v,l)$-cover-free with high probability. 

Now we define the witness algorithm of an $(r,v,\ell)$-cover-free set family.
\begin{definition}\label{def:wit_rvl}
Suppose $\cJ=\{F_1,F_2,\dotsc,F_N\}$  is an \emph{$(r,v,\ell)$-cover-free} family over $M$. We say $\cJ$ can be witnessed in time $T_w$ if there is a witness algorithm $W$ such that for any input  $u\in [N]$, $U_1,U_2,\dotsc,U_r\subset [N]$, where $u\in U_i$ and $\abs{U_i}\leq v$ for all $i\in[r]$, $W(u,U_1,U_2,\dotsc,U_r)$ outputs in time $T_w$ a witness element $x$ of $F_u$ such that for all $i\in [r]$, $x\in F_s$ holds for at most $\ell$ values of $s\in U_i$. 
% We say that $\cJ$ can be checked in time $T_c$, if there is a checking  algorithm $C$ such that for any $i\in [N]$ and $x\in[t]$,  $C$ can check whether $x\in F_i$ in time $T_c.$
\end{definition}

Next, we apply Lemma~\ref{lem:rvlexist} to the case where $r,v,N$ are polynomials in $n$.

\begin{lemma}\label{lem:hyp}
   Suppose $r(n),v(n),N(n)=\poly(n)$, and $\ell$ is a positive integer. There exists an $(r(n),v(n),\ell)$-cover-free set family of size $N(n)$ over $[t(n)]$ with a $\poly(n)$-time witness algorithm, where $t(n)=O(r(n)^{1+1/{\ell}}v(n)^{2}\log N(n))$.
\end{lemma}

\begin{proof}
The existence of the $(r(n),v(n),\ell)$-cover-free set family is proved in Lemma~\ref{lem:rvlexist}. Suppose the family is $$\cJ=\{F_1,F_2,\dotsc,F_{N(n)}\}.$$
For $u\in [N(n)]$, $U_1,U_2,\dotsc,U_r\subset [N(n)]$, where $u\in U_i$ and $\abs{U_i}\leq v(n)$ for all $i\in[r(n)]$, we let $W(u,U_1,U_2,\dotsc,U_r)$ brute force all the $x\in F_u$. Since $r(n),v(n),N(n), t(n)$ are all $O(\poly(n))$,  the witness algorithm runs in time $O(\poly(n))$.  
\end{proof}

\subsection{Labeling using Cover-free Set Families}

 We first discuss the relationship between colorings and labelings, starting with one more definition. 
% \begin{definition}
%     A hypergraph $H$ is $T$-explicit if 
%     \begin{enumerate}
%         \item there is an order and an algorithm $A$ such that given a vertex index $v$ and $i\in [r(H)]$, $A$ can output the $i$-th edge that contains $v$ (in this order) in time $T$.
%         \item there is an order and an algorithm $B$ such that given an edge index $E$ and $i\in [v(H)]$, $B$ can output the $i$-th vertex in $E$ (in this order) in time $T$.
%     \end{enumerate}
    
% \end{definition}

\begin{definition}
    A hypergraph $H$ is $T$-explicit if  the following holds.
    \begin{enumerate}
        \item There is an ordering of edges and an algorithm  such that given a vertex index $v$ and $1\le i\le r(H)$, the algorithm can output the $i$-th edge that contains $v$ (in this ordering) in time $T$.
        \item There is an ordering of vertices and an algorithm  such that given an edge index $E$ and $1\le i\le v(H)$, the algorithm can output the $i$-th vertex in $E$ (in this ordering) in time $T$.
    \end{enumerate}
\end{definition}

 Note that if $\Phi:V(H)\to A$ is a coloring of $G(H)$, then $\Phi$ is a $1$-labeling of $H$. If $H$ is $T_e$-explicit, and a coloring of $G(H)$ can be individually computed in time $T$, then as a $1$-labeling of $H$, it can be individually computed in time $T$ and decoded in time $v(H)(T_e+T)$. For $a\in A$ and an edge $E$ of $H$, $\Dec(a,E)$ just compute $\Phi(x)$ for all $x\in E$, and output those $x$ satisfying $\Phi(x)=a$.
 
Next we show that one can construct a new 
$\ell\ell_1$-labeling of $H$ from an existing $\ell_1$-labeling through an $(r(H),v(H),\ell)$-cover-free set family.

\begin{lemma}\label{lem:llabeling}
Suppose $H$ is a $T_e$-explicit hypergraph with $N$ vertices, and  $\cJ=\{F_1,F_2,\dotsc,F_N\}$ is an $(r(H),v(H),\ell)$-cover-free set family of size $ N$ over ground set $[t]$ that has a witness algorithm  that runs in time $T_w$. If $\Phi:V(H)\to [N]$  is an existing $\ell_1$-labeling that can be individually computed in time $T_1$ and decoded in time $T_2$,  then we can get a new $\ell\ell_1$-labeling $\Phi^*:V(H)\rightarrow [t]$. Moreover, $\Phi^*$ can be individually computed in time $T_w+(2T_e+T_1)r(H)v(H)$ and decoded in time $v(H)(T_e+T_1)+\ell T_2$.
\end{lemma}

\begin{proof} 
 For a vertex $u$ of $H$, suppose $E_1,E_2,\dotsc,E_{r'}$ ($r'\leq r(H)$) are the edges that contain $u$. Since $\cJ$ is an $(r(H),v(H),\ell)$-cover-free set family with witness algorithm $W$, $W(\Phi(u),\Phi(E_1),\Phi(E_2),\dotsc,\Phi(E_r'),\Phi(E_r'),\dotsc,\Phi(E_r'))$ can output  an $x \in [t]$ satisfying that $x\in F_{\Phi(u)}$ and $x\in F_{j}$ holds for at most $\ell$ choices of $j\in\Phi(E_i)$ for $i\in [r']$. We define $\Phi^*(u)$ to be $x$. Note that it takes at most $(2T_e+T_1)r(H)v(H)$ time to compute $\Phi(u),\Phi(E_1),\Phi(E_2),\dotsc,\Phi(E_r')$. So it takes $T_w+(2T_e+T_1)r(H)v(H)$ to compute $\Phi^*$.

We describe the $\ell\ell_1$-decoder of $\Phi^*$. 
    For any $b\in [t]$ and edge $E$ of $H$, 
    let $$D(b,E)=\{j\in \Phi(E):b\in F_{j}\},$$ 
    and
    $$\Dec_{\Phi^*}(b,E)=\begin{cases}
        \Dec_{\Phi}(D(b,E),E) &\abs{D(b,E)}\leq \ell,\\
        \emptyset &\abs{D(b,E)}> \ell.
    \end{cases}$$
    Then $\abs{\Dec_{\Phi^*}(b,E)}\leq \ell\ell_1$. 
    For any vertex $u$ and edge $E$ of $H$, where $u\in E$, by the definition of $\Phi^*(u)$, we have $\Phi(u)\in D(\Phi^*(u),E)$ and $D(\Phi^*(u),E)\leq \ell$. Thus $$u\in \Dec_{\Phi}(D(\Phi^*(u),E),E)=\Dec_{\Phi^*}(\Phi^*(u),E).$$ So, $\Dec_{\Phi^*}$ is an  $\ell\ell_1$-decoder of $\Phi^*$, and $\Phi^*$ is an $\ell\ell_1$-labeling. Moreover, $D(b,E)$ can be computed in time $v(H)(T_e+T_1)$. So $\Dec_{\Phi^*}$ can be computed in time $v(H)(T_e+T_1)+\ell T_2$.
\end{proof}

\begin{lemma}\label{lem:1/l}
    Suppose $H$ has $2^{p(n)}$ vertices, and $r(H)=r(n),v(H)=v(n)$, where $p(n),r(n),v(n)$ are $O(\poly(n))$. Suppose $H$ is $\poly(n)$-explicit.  Then there exists an $\ell$-labeling $\Phi:V(H)\rightarrow\Tilde{O}(r(n)^{1+1/{\ell}}v(n)^2)$, where $\Phi$ can be individually computed and decoded in time $\poly(n)$.
\end{lemma}

\begin{proof}
Note that $\Delta(G(H))\leq r(n)v(n)$, and $G(H)$ is $r(n)v(n)$-explicit.    By Lemma~\ref{lem:deltaG2}, one can construct a coloring $\Phi_1:G(H)\to [O((r(n)v(n)p(n))^2)]$, which can be individually computed in time $\poly(n)$. So $\Phi_1:V(H)\to [O((r(n)v(n)p(n))^2)]$ is a $1$-labeling that can be individually computed and decoded in time $\poly(n)$. By Lemma~\ref{lem:llabeling}, viewing $\Phi_1$ as the existing $1$-labeling, using the $(r(n),v(n),\ell)$-cover-free set family in Lemma~\ref{lem:hyp} with $N(n)=O((r(n)v(n)p(n))^2)$, one can construct a new $\ell$-labeling $\Phi:V(H)\to [\Tilde{O}(r(n)^{1+1/{\ell}}v(n)^2)]$. Moreover, $\Phi$ can be individually computed and decoded in time $\poly(n)$.
\end{proof}

\subsection{Construction of List-Decodable Codes}
The following theorem is a direct result of \Cref{lem:1/l}.
\thmlist*

Note that compared to our  explicit construction of unique decodable codes, the  construction of list decodable codes is based on the existential result of the $(r,v,l)$-cover-free set family.
Now, to compare Theorem~\ref{thm:list} with Theorem~\ref{thm:GV}, we can apply Theorem~\ref{thm:list} to $H_n$, where $\cE$  represents a constant number of edits (and where each edit can be an insertion, deletion, or substitution). For example, if we take $\cE$ to be $k$ edits, where $k$ is a constant, then we can get the following corollary.

\begin{corollary}\label{cor:keditl}
 Suppose $k$ is a constant and $q(n)=O(\poly(n))$.    There exists a function \added{$$\Phi:\Sigma_{q(n)}^n\rightarrow [2^{(3+1/{\ell})k\log n+(3+1/{\ell})k\log q(n)+O(\log\log n)}]$$} that can correct $k$ edits with list size $\ell$. Specifically, if $y$ is obtained by $k$ edits from $x$, then given $y$ and $\Phi(x)$ one can recover a list of size $\ell$ that contains $x$. Moreover, computing $\Phi$ and the recovery process takes $\poly(n)$ time.
\end{corollary}
\begin{proof}
    Let $\cE$ be $k$ edits, then $H_{n}$ has $q(n)^n=2^{O(\poly( n))}$ vertices.  
$r(H_{n})=O(n^{2k}q(n)^{2k})=O(\poly(n))$.  $v(H_{n})=O(n^{2k}q(n)^{2k})=O(\poly(n))$. Moreover, $H_{n}$ is $\poly(n)$-explicit. We apply Theorem~\ref{thm:list} and complete the proof.
\end{proof}

\section{Incremental Synchronization}\label{sec:fs}
In this section, we introduce another application of the graph coloring framework, which we call `incremental synchronization.' Consider the following file synchronization scenario, where Alice has a string $x$, and Bob
has a string $y$. Suppose $x$ and $y$ differ by a number of edits. Alice wants to send Bob a syndrome $\Phi(x)$ of $x$ such that Bob can determine $x$ using $\Phi(x)$ and $y$. If Alice and Bob know the number of edits between $x$ and $y$, say $k$, then Alice can simply send the syndrome in Corollary~\ref{cor:kedit} to Bob, denoted here as $\Phi_{k}(x)$, which has a length of about $4k\log n$. However, in practice,  Alice and Bob may not know $k$. Instead, they may know the distribution of $k$. We consider a situation where $k$ is known to be small, say $k\le a$, with probability $1-p$, and large but still bounded, say $a<k\le b$, with probability $p$, for $a<b$. In this scenario, we are interested in synchronization protocols with the smallest communication load. 

In a ``naive'' approach, Alice sends $\Phi_{b}(x)$ to Bob, which costs $\abs{\Phi_b(x)}\simeq 4b\log n$ (the length of $\Phi_b(x)$). On the other hand, in terms of lower bounds, assuming the use of codes in \Cref{cor:kedit}, \added{and assuming that no shorter syndrome than $\Phi$  of \Cref{cor:kedit} exists}, we need at least $(1-p)\abs{\Phi_a(x)}+p\abs{\Phi_b(x)}\simeq 4(a+p(b-a))\log n$ bits on average, as Alice can do no better than when she has access to an Oracle with knowledge of the distance.

\begin{figure}
    \centering
    \def\a{.5}
    \def\b{2}

    \begin{tikzpicture}
      \begin{axis}[
        width=2.5in,
        height=2in,
        domain=0:1,
        samples=10,
        ylabel={$B/(4\log n)$},
        ylabel style={yshift=-0pt,font=\small},
        ymin=0,                         
        ytick={\a,(\a+\b)/2,\b,\a/4+5*\b/4,\a/2+3*\b/2},                  
        yticklabels={$a$,$\frac{a+b}2$,$b$,$\tfrac{a + 5b}{4}$,$\tfrac{a + 3b}{2}$},            
        xtick={0,1},                  
        xticklabels={0,1},
        extra x ticks={0.5},
        extra x tick labels={$p$},
        extra x tick style={tick style={draw=none},yshift=-10pt},
        clip=false,                     % allow legend outside
        legend cell align=left,
        legend style={
          font=\scriptsize,
          at={(axis description cs:1.02,0.5)},  % slightly to the right, vertically centered
          anchor=west,
          draw=none,
          fill=none,
          cells={anchor=west}
        },
        thick
      ]
        % Oracle lower bound — solid
        \addplot[black,solid] { \a + x*(\b - \a) };
        \addlegendentry{Oracle LB}

        % Naive protocol — dashed
        \addplot[green!60!black,dashed] { \b };
        \addlegendentry{Naive protocol}

        % (a+b)/2 + p b — dotted
        \addplot[blue,dash dot] { (\a + \b)/2 + x*\b };
        \addlegendentry{Test-and-fallback}

        % Inc. Synch. — dash-dot
        \addplot[red,solid,mark=o] { (\a + \b)/2 + x*(\b - (\a + \b)/4) };
        \addlegendentry{Inc. Synch.}
      \end{axis}
    \end{tikzpicture}

    \caption{The number of bits $B$ that Alice needs to send for synchronization based on the Oracle lower bound and three protocols (divided by $4\log n$ as a function of $p$, where the distribution of the edit distance is $\Pr(d_E(x,y)=a)=1-p,\Pr(d_E(x,y)=b)=p$.}
    \label{fig:inc-synch}
\end{figure}
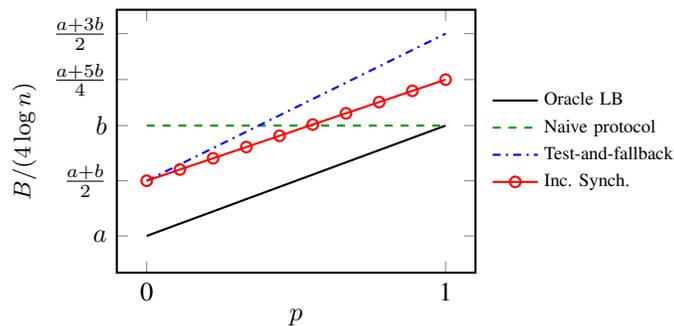

We propose two other approaches to this problem and compare their average communication costs.  In the ``test-and-fallback'' approach, Alice first sends $\Phi_{\ceil{\frac{a+b}{2}}}(x)$ to Bob. Note that $d_E(x,y)\leq a$ if and only if there exists a $z\in B_{a}(y)$ such that $\Phi_{\ceil{\frac{a+b}{2}}}(z)=\Phi_{\ceil{\frac{a+b}{2}}}(x)$. The forward direction is straightforward. We now turn to the backward direction, which needs further explanation. Note that $d_E(z,y)\leq a$ and $d_E(x,y)\leq b$. So $d_E(x,z)\leq a+b$. Since $\Phi_{\ceil{\frac{a+b}{2}}}(z)=\Phi_{\ceil{\frac{a+b}{2}}}(x)$, we have $x=z$. Since $z\in B_{a}(y)$, we have $d_E(x,y)\leq a$. So if there exists such $z$, then Bob knows that $x=z$. Otherwise, Alice sends $\Phi_{b}(x)$ to Bob. The expected communication from Alice to Bob is $\abs{\Phi_{\ceil{\frac{a+b}{2}}}(x)}+p\abs{\Phi_b(x)}$. (We ignore the 1-bit message sent from Bob to Alice.) There exists $0<p_0<1$ such that the second approach has a lower expected communication load than the first when $p<p_0$. An example for the special case of $\Pr(d_E(x,y)=a)=1-p, \Pr(d_E(x,y)=b)=p$, is given in \Cref{fig:inc-synch}.

The test-and-fallback approach leads to an intriguing question. Suppose after the first message sent by Alice, Bob fails to determine $x$. At this point, even though Bob does not know $x$, he has some information about it as he knows $\Phi_{\ceil{\frac{a+b}{2}}}(x)$. Can we take advantage of this fact and send a message shorter than $\abs{\Phi_b(x)}$? The next theorem, provides an affirmative answer and leads to a protocol that we call ``incremental synchronization.'' The key idea of the proof is that using the prior information, we construct a graph with a smaller maximum degree than the graph for correcting $b$ edits.%, notice that in the last step of the second approach,  Alice may not need to send $\Phi_{b}(x)$ since Bob already has $\Phi_{\ceil{\frac{a+b}{2}}}(x)$. Specifically, we have the following theorem.
% \begin{reptheorem}{thm:incremental}
% Let $b>a$. 
% Suppose $h_a$ is a syndrome that can correct $a$ edits. If $y$ is obtained from $x$ by $b$ edits, then given $h_a(x)$ one can construct a syndrome $h_{b|a}:\Sigma_2^n \rightarrow [2^{(4b-2a)\log n+O(\log \log n)}]$ such that $x$ can be recovered given $y$, $h_a(x)$, and $h_{b|a}(x)$. Moreover, if $h_a$  can be computed in $\poly(n)$ time, then so does $h_{b|a}$.
% \end{reptheorem}
\thmincremental*

\begin{proof}
    Let  \added{$C:=\{z\in \Sigma_2^n:\Phi_a(z)=\Phi_a(x)\}$}. Then $C$ is a code with edit distance at least $2a+1$. Let $\cE_b$ be the channel that applies at most $b$ edits, and let $G$ be the confusion graph for this channel but with inputs being the set $C$. Note that $(u,v)$ is an edge if the edit distance between $u$ and $v$ is at most $2b$. So the degree of a vertex $u$ in $G$ is at most $\abs{B_{2b}(u)\cap C}$. To bound the maximum degree of $G$, we need Lemma~\ref{lem:distance}, which is a generalization of Lemma 4 in \cite{song22systematic}. By Lemma~\ref{lem:distance}, with $D=2b$ and $d=2a+1$, we have $\Delta (G)\leq O(n^{2b-a})$. 

Now, if $\Phi_a$  can be computed in $\poly(n)$ time, then $G$ is $\poly(n)$-explicit. Thus, by Lemma~\ref{lem:deltaG2}.\ref{lem:deltaG22}, one can construct a coloring $\Phi:V(G)\rightarrow [\Tilde{O}(n^{4b-2a})]$ that can be individually computed in $\poly(n)$ time. This $\Phi$ is our desired $\Phi_{b|a}$. 
\end{proof} 
Note that the maximum degree of the graph $G$ defined in the proof, i.e., $O(n^{2b-a})$, is lower than the maximum degree of the original confusion graph for $\le b$ edits with inputs $\Sigma_2^n$, i.e., $O(n^{2b})$. This reduction is the result of limiting the vertex set to vertices with a given value of $\Phi_a$.

\begin{lemma}\label{lem:distance}
    Suppose that $C$ is a code with edit distance $d$, then for any 
$c\in C$ and $D\geq \floor{\frac{d-1}{2}}$, we have $\abs{B_D(c)\cap C}\leq \abs{B_{D-\floor{\frac{d-1}{2}}}(c)}$. 
\end{lemma}

\begin{proof}
  We can construct an injection $f$ from   $B_D(c)\cap C$ to $B_{D-\floor{\frac{d-1}{2}}}(c)$ as follows. For any $c'\in B_D(c)\cap C$, there is a sequence of at most $D$ steps of  edits that can change $c'$ to $c$. We let $f(c')$ be the one at the $\floor{\frac{d-1}{2}}$-th step  of this sequence.  Note that $f(c')\in B_{D-\floor{\frac{d-1}{2}}}(c)$. Moreover, if $c_1', c_2'\in B_D(c)\cap C$ and $f(c_1')=f(c_2')$, then $d_E(c_1',c_2')\leq d-1$. Since $C$ has distance $d$, we have $c_1'=c_2'$. Thus $f$ is an injection from $B_D(c)\cap C$ to $B_{D-\floor{\frac{d-1}{2}}}(c)$.
\end{proof}

Suppose Alice has $x$, and Bob has $y$. $\Pr[d_E(x,y)\leq a]=1-p$,  $\Pr[a<d_E(x,y)\leq b]=p$.
We summarize the incremental synchronization protocol as follows. 
\begin{enumerate}
    \item Alice: Sends $\Phi_{\ceil{\frac{a+b}{2}}}(x)$ to Bob.
    \item Bob: If there exists $ z\in B_{a}(y)$ such that  $\Phi_{\ceil{\frac{a+b}{2}}}(z)=\Phi_{\ceil{\frac{a+b}{2}}}(x)$, then Bob outputs $z$, sends $1$ to Alice, and terminates the algorithm. Otherwise Bob sends $0$ to Alice. 
    \item Alice: If receives $0$, then sends $\Phi_{b|\ceil{\frac{a+b}{2}}}(x)$ to Bob.
    \item Bob: Searches for the unique $z\in B_{b}(y)$ such that $\Phi_{\ceil{\frac{a+b}{2}}}(z)=\Phi_{\ceil{\frac{a+b}{2}}}(x)$ and $\Phi_{b|\ceil{\frac{a+b}{2}}}(z)=\Phi_{b|\ceil{\frac{a+b}{2}}}(x)$. Outputs $z$.
\end{enumerate}
The expected communication from Alice to Bob is 
\begin{align*}
    \abs{\Phi_{\ceil{\frac{a+b}{2}}}(x)}+p\abs{\Phi_{b|\ceil{\frac{a+b}{2}}}(x)},
\end{align*}
which is 
\begin{align*}
  \left(4 \ceil{\frac{a+b}{2}}+p\left(4b-2\ceil{\frac{a+b}{2}}\right)\right)\log n+O(\log\log n)
\end{align*}
by applying Corollary~\ref{cor:kedit} and Theorem~\ref{thm:incremental}.
The incremental synchronization scheme has a lower expected communication load than the test-and-fallback approach. Moreover, there exists $p_1>p_0$ such that incremental synchronization has lower communication than the naive approach when $p<p_1$. These are shown for the special case of $\Pr(d_E(x,y)=a)=1-p,\Pr(d_E(x,y)=b)=p$ in \Cref{fig:inc-synch}.

\section{Codes Correcting a Constant Number of Substring 
Edits}\label{sec:multisubstring}
Recall that an $l$-substring edit is the operation of replacing a substring $u$ of $x$ with another string $v$, where $\abs{u},\abs{v}\le l$.
In this section, we will study codes correcting $k$ $l$-substring edits, where $k$ is a constant and $l$ may grow with $n$. First, we derive two bounds on the redundancy. Second, we construct codes using Theorem~\ref{thm:GV} for $l=O(\log n)$.  Third, to extend Theorem~\ref{thm:GV} to $l=\omega(\log n)$, we first consider a simplified case where the substring edits is a  burst of deletions. Then, we generalize this method to correct a constant number of substring edits.
\subsection{Bounds}
In this subsection, we will derive a lower bound and an upper (existential) bound on the redundancy of codes correcting multiple substring edits, in \Cref{lem:hamming} and \Cref{lem:existential}, respectively.
%\subsubsection{Lower bounds on redundancy}
\begin{lemma}[Hamming bound]\label{lem:hamming}
Let $n,k,l$ be positive integers.     A code correcting $k$ $l$-substring edits has redundancy at least $\max\{\log \binom{n}{k},kl\}$.
\end{lemma}
\begin{proof}
    Let $\cE$ be $k$ $l$-substring edits.   We  show that for any $x\in \Sigma_2^n$, we have $\abs{\cE(x)}\geq \max\{\binom{n}{k},2^{kl}\}$. 
 First, $\cE(x)$ contains all sequences that can be obtained by substituting $k$ bits in $x$. Note that there are at least $\binom{n}{k}$ such sequences. So $\abs{\cE(x)}\geq \binom{n}{k}$. Similarly, $\cE(x)$ contains all sequences that can be obtained by substituting the first $kl$ bits of $x$. There are $2^{kl}$ such sequences. Thus $\abs{\cE(x)}\geq 2^{kl}.$

 Suppose that $C$ is a code correcting $\cE$, then for any $x_1\neq x_2\in C$, $\cE(x_1)\cap \cE(x_2)=\emptyset$. Thus $\abs{C}\leq \frac{2^n}{\max\{\binom{n}{k},2^{kl}\}}$. Hence, the redundancy is at least $\max\{\log \binom{n}{k},kl\}$.
\end{proof}

%\subsubsection{Existential Bounds}\label{subsubsec:existential}

\begin{lemma}[Gilbert-Varshamov bound]\label{lem:existential}
    There exist codes correcting $k$ $l$-substring edits with redundancy $2k\log n+2kl+4k\log (l+1)+O(1)$.
\end{lemma}
\begin{proof}
Let $\cE$ be   $k$ $l$-substring edits.  We first bound $\Delta(G_n)$. To this end, we fix a vertex $u$ of $G_n$ and count the number of its neighbors. For a vertex $v$ to be a neighbor of $u$, there exists some $y$ such that $y\in\cE(u)$ and $y\in\cE(v)$. 
This implies that $v$ is obtained by deleting $2k$ bursts of length at most $l$ and inserting $2k$ bursts of length at most $l$ at same positions. There are at most  $(n(l+1))^{2k}$ ways to delete $2k$ bursts of length at most $l$, and  there are at most $(l+1)^{2k}$ ways to choose insertion lengths of the $2k$ bursts insertions of length at most $l$, and there are $(2^{l})^{2k}$ choices for the insertion content. Therefore, 
$\Delta(G_n)=O((n(l+1)^22^l)^{2k})$. So by a greedy argument, there is a coloring $\Phi:G_n\to[O((n(l+1)^22^l)^{2k})]$. Thus, there exist a code correcting $k$ $l$-substring edits,  $\Phi^{-1}(c)$ for some color $c$, with redundancy $2k\log n+2kl+4k\log (l+1)+O(1)$.
\end{proof}

\subsection{Codes Correcting a Constant Number of Short Substring Edits}

Now we apply Theorem~\ref{thm:GV} to $G_n$, where $\cE$ is   $k$ $l$-substring edits,  $k$ is a constant, and $l=O(\log n)$. By Theorem~\ref{thm:GV}, we can get the following corollary.

\begin{corollary}\label{cor:l=logn}
    Suppose $k$ is a positive integer  and $l=O(\log n)$. There is a function  \[\Phi:\Sigma_2^n\rightarrow [2^{4k(\log n+l+2\log (l+1))+O(\log\log n)}]\] such that  $\{(x,\Phi(x)):x\in\Sigma_2^n\}$ can correct $k$ $l$-substring edits (\added{for synchronization channel}), with $\poly(n)$ time encoding and decoding complexity.
\end{corollary}
\begin{proof}
    Let $\cE$ be $k$ $l$-substring edits, then $G_n$ has $2^n$ vertices. The argument in the proof of  \Cref{lem:existential} shows that $\Delta(G_n)=O((n(l+1)^22^l)^{2k})$.   Moreover, the  argument in the proof of \Cref{lem:existential} gives us a way to determine the $i$-th neighbor of $u$ for $i\in [\Delta(G_n)]$. So 
 $G_n$ is $O((n(l+1)^22^l)^{2k})$-explicit. When $l=O(\log n)$, $\Delta(G_n)=\poly(n)$ and  $G_n$ is $\poly(n)$-explicit. The desired result follows from Theorem~\ref{thm:GV}.%, we can construct a function $\Phi:\Sigma_2^n\rightarrow [\tilde{O}((n(l+1)^22^l)^{4k})]=[2^{4k(\log n+l+2\log (l+1))+O(\log\log n)}]$, which can be individually computed in time $\poly(n)$. $\Phi$ is our desired function. Suppose the corrupted codeword is $y$. To decode, the decoder computes $\Phi(z)$ for all $z$ satisfying $y\in\cE(z)$ and choose the $z$ satisfying $\Phi(z)=\Phi(x)$ as the output. Since $l=O(\log n)$, there are $\poly(n)$ such $z$, and $\Phi(z)$ can be computed in $\poly(n)$ time. Thus, the decoder runs in $\poly(n)$ time. 
\end{proof}

\subsection{Codes Correcting a Long Burst of Deletions of Variable Length with Asymptotically Optimal Redundancy}\label{sec:long}
Note that the result in \Cref{cor:l=logn} only works for $l=O(\log n)$. When $l\neq O(\log n)$, the maximum degree of $G_n$ in the proof of \Cref{cor:l=logn} may not be a polynomial in $n$. Thus the encoder and decoder may not be polynomial-time. To construct codes that work for $l= \omega( \log n)$, we first consider a simple case, where the error model is one burst of deletions of length~$l$.

In this subsection, we use Lemma~\ref{lem:deltaG2} to  construct codes correcting a burst of $l$ deletions with redundancy $l+O(\log n)$, which is asymptotically optimal when $l=\omega( \log n)$. Moreover, our codes have polynomial-time encoders and decoders. 
Specifically, let $\varepsilon$ be a burst of at most $l$  deletions. We want to construct a syndrome $\Phi$, such that if $y\in\cE(x)$, then one can recover $x$ with $y$ and $\Phi(x)$. 

If we use the method in the last subsection, we will find that the degree of the graph $G_n$ exceeds $\poly(n)$, therefore we cannot apply  Theorem~\ref{thm:GV} directly to construct a coloring that is computable in $\poly(n)$.
Note that when $l$ is a constant,  $\Phi(x)$ should mainly help recover  the deletion position, and this needs about $\log n$ redundancy. While when $l=\omega( \log n)$, $\Phi(x)$ should mainly help recover  the deleted string, which needs about $l$ redundancy. This is the key difference between the case when $l$ is a constant and when $l=\omega( \log n)$. 
By applying Linial's algorithm, one can `separate' recovering the deleted string and recovering the deletion position. 

Our method is inspired by the precoding method in \cite{song22systematic}. Specifically, we define 
\begin{align*}
    \phi_1(x)=\left(\bigoplus_{i\equiv_l 1}x_i,\bigoplus_{i\equiv_l 2}x_i,\dotsc, \bigoplus_{i\equiv_l l}x_i\right),
\end{align*}
where $i\equiv_l j$ denotes $i\equiv j \mod l$.
Let $G$ be a graph whose vertex set is $\Sigma_2^n$, $(u,v)$ is an edge if $\phi_1(u)=\phi_1(v)$ and there exists some $y$ such that $y\in\cE(u)$ and $y\in\cE(v)$.

We first bound the maximum degree of $G$. To this end, we fix a vertex $u$ of $G$ and count the number of its neighbors. For a vertex $v$ to be a neighbor of $u$, there exists some $y$ such that $y\in\cE(u)$ and $y\in\cE(v)$. There are $O((l+1)n)$ choices for $y$ ($l+1$ for the possible  number of deletion bits and $n$ for the possible deletion positions). For each $y$, there are $O(n)$ choices for $v$. This is because there are $O(n)$ choices for the insertion position to get $v$ from $y$, and once the insertion position is fixed, since $\phi_1(u)=\phi_1(v)$, $v$ is determined. Thus $\Delta(G)=O((l+1)n^2)$. Moreover, the above argument gives us a way to determine the $i$-th neighbor of $u$ for $i\in [\Delta(G)]$. So 
 $G$ is $O((l+1)n^2)$-explicit. Therefore, by Lemma~\ref{lem:deltaG2}.\ref{lem:deltaG22}), one can construct a coloring $\phi_2:V(G)\rightarrow [\Tilde{O}((l+1)^2n^4)]$, which can be computed in $\poly(n)$ time.

 \begin{theorem}
    Let $l$ be a positive integer. 
     If $y$ is obtained from $x$ by a burst of at most $l$ deletions, then one can recover $x$ given $y$, $\phi_1(x)$, and $\phi_2(x)$ in $\poly(n)$ time.
 \end{theorem}

 \begin{proof}
     Given $y$, we recover $x$ by inserting $\abs{x}-\abs{y}$ consecutive bits. There are $\abs{y}+1$ possible insertion positions in total. For each insertion position $i\in [\abs{y}+1]$, there is only one possible $z_i\in \Sigma_2^n$ (the string after insertion from $y$) satisfying $\phi_1(z_i)=\phi_1(x)$ (one can solve  $z_i$ in $O(l)$ time). Therefore $x\in \{z_i:i\in [\abs{y}+1]\}$. Moreover, note that $\{z_i:i\in [\abs{y}+1]\}$ forms a clique of $G$. So there is a unique $z_i$ such that $\phi_2(z_i)=\phi_2(x)$ since $\phi_2$ is a coloring of $G$. This $z_i$ is $x$. The whole process runs in $\poly(n)$ time since $\phi_2$ can be computed in $\poly(n)$ time.
 \end{proof}
Let 
\begin{align*}
    C:=\{x\in \Sigma_2^n:\phi_1(x)=a, \phi_2(x)=b\}.
\end{align*}
Then $C$ can correct a burst of at most $l$ deletions. The redundancy of $C$ for some $a$ and $b$ is $l+O(\log n)$, which is asymptotically optimal when $l=\omega( \log n)$.

\subsection{Codes Correcting a Constant Number of Long Substring Edits}
In this subsection, we study codes correcting $k$ $l$-substring edits for $l=\omega(\log n)$. Note that the $\phi_1$ in the last subsection serves as an erasure correcting syndrome.
Inspired by the last subsection, we will use a new  $\phi_1$, a syndrome  correcting $2k$ bursts of erasures of length $l$, which we can  construct using Reed-Solomon codes. Specifically, let $f_{q,m,\kappa}$ be such that  $\{(u,f_{q,m,\kappa}(u)):u\in \mathbb{F}_{q}^m\}=\text{RS}_q(m+\kappa,\kappa)$, where $q\geq m+\kappa$. Let $\phi_1(x):=f_{2^l,\frac{n}{2^l},4k}(\tilde{x})$, where $\tilde{x}$ is obtained by dividing the length $n$ string $x$ into blocks of length $l$, and view each block as an element in $\mathbb{F}_{2^l}$. Since $l=\omega(\log n)$, we have $2^l>\frac{n}{2^l}+4k$ when $n$ is sufficiently large. Note that $\phi_1$ is a syndrome that can correct $2k$ bursts of erasures of length $l$. Let $G$ be a graph whose vertex set is $\Sigma_2^n$, $(u,v)$ is an edge if $\phi_1(u)=\phi_1(v)$ and there exists some $y$ such that $y\in\cE(u)$ and $y\in\cE(v)$.

We first bound the maximum degree of $G$. To this end, we fix a vertex $u$ of $G$ and count the number of its neighbors. For a vertex $v$ to be a neighbor of $u$, there exists some $y$ such that $y\in\cE(u)$ and $y\in\cE(v)$. 
This implies that $v$ is obtained by deleting $2k$ bursts of length at most $l$ and inserting $2k$ bursts of length at most $l$ at same positions. There are $(n(l+1))^{2k}$ ways to delete $2k$ bursts of length at most $l$, and  there are $(l+1)^{2k}$ ways to choose insertion lengths of the $2k$ bursts insertions of length at most $l$. Having determined these, $v$ is determined since $\phi_1(v)=\phi_1(u)$. Thus, $\Delta(G)\leq(n(l+1)^2)^{2k}$.

 Moreover, the above argument gives us a way to determine the $i$-th neighbor of $u$ for $i\in [\Delta(G)]$. So 
 $G$ is $O((n(l+1)^2)^{2k})$-explicit  (the decoding of Reed-Solomon codes takes $\poly(n)$ time). Therefore, by Lemma~\ref{lem:deltaG2}.\ref{lem:deltaG22}), one can construct a coloring $\phi_2:V(G)\rightarrow [\Tilde{O}((l+1)^{8k}n^{4k})]$, which can be computed in $\poly(n)$ time.

\begin{theorem}
  Let $k$ be a positive integer and  $l=\omega(\log n)$.
  If $y$ is obtained from $x$ by $k$ $l$-substring edits, then one can recover $x$ given $y$, $\phi_1(x)$, and $\phi_2(x)$ in $\poly(n)$ time.
 \end{theorem}

 \begin{proof}
     Given $y$, we recover $x$ by performing $k$ $l$-substring edits on $y$. There are  $(\abs{y}+1)^k$ substring edit  positions in total. For each substring edit positions $\boldsymbol{i}=(i_1,i_2,\dotsc,i_k)$, there are $(l+1)^k$ possible deletion lengths $\boldsymbol{l}=(l_1,l_2,\dotsc,l_k)$ and $(l+1)^k$ possible insertion lengths $\boldsymbol{L}=(L_1,L_2,\dotsc,L_k)$. After having determined $\boldsymbol{i}$, $\boldsymbol{l}$ and $\boldsymbol{L}$, 
     there
     is at most one possible $z_{\boldsymbol{i},\boldsymbol{l},\boldsymbol{L}}\in \Sigma_2^n$  satisfying $\phi_1(z_{\boldsymbol{i},\boldsymbol{l},\boldsymbol{L}})=\phi_1(x)$ since $\phi_1$ can correct $k$ burst of $l$ erasures (one can solve $z_{\boldsymbol{i},\boldsymbol{l},\boldsymbol{L}}$ in $\poly(n)$  time). Therefore $x\in \{z_{\boldsymbol{i},\boldsymbol{l},\boldsymbol{L}}:\boldsymbol{i}\in [\abs{y}+1]^k,\boldsymbol{l}\in [l+1]^k,\boldsymbol{L}\in[l+1]^k\}$. 
     
     Moreover, note that $\{z_{\boldsymbol{i},\boldsymbol{l},\boldsymbol{L}}:\boldsymbol{i}\in [\abs{y}+1]^k,\boldsymbol{l}\in [l+1]^k,\boldsymbol{L}\in[l+1]^k\}$ forms a clique of $G$. So there is a unique $z_{\boldsymbol{i},\boldsymbol{l},\boldsymbol{L}}$ such that $\phi_2(z_{\boldsymbol{i},\boldsymbol{l},\boldsymbol{L}})=\phi_2(x)$ since $\phi_2$ is a coloring of $G$. This $z_{\boldsymbol{i},\boldsymbol{l},\boldsymbol{L}}$ is $x$. The whole process runs in $\poly(n)$ time since 
 $\phi_1$ and $\phi_2$ can be computed in $\poly(n)$ time and $\{z_{\boldsymbol{i},\boldsymbol{l},\boldsymbol{L}}:\boldsymbol{i}\in [\abs{y}+1]^k,\boldsymbol{l}\in [l+1]^k,\boldsymbol{L}\in[l+1]^k\}$ has $\poly(n)$ size.
 \end{proof}

Let 
\begin{align*}
    C:=\{x\in \Sigma_2^n:\phi_1(x)=a, \phi_2(x)=b\}.
\end{align*}
Then $C$ can correct $k$ $l$-substring edits. The redundancy of $\phi_1$ is $4kl$, and the redundancy of $\phi_2$ is $4k\log n+4k\log (l+1)+O(\log\log n)$.
Thus, 
the redundancy of $C$ for some $a$ and $b$ is $4kl+4k\log n+8k\log (l+1)+O(\log\log n)$,  which is approximately twice the Gilbert-Varshamov bound and 8 times the Hamming bound.

\section{Conclusion}
In this work, we proposed a framework for constructing error-correcting codes for a constant number of errors of any type, as long as the confusion graph for the channel satisfies mild conditions. This framework is based on distributed graph coloring, powered by Linial's algorithm. The proposed method achieves redundancy twice the GV bound, where in contrast to prior approaches, there is no need to start with a code with reasonably low redundancy. Furthermore, the proposed distributed graph coloring approach is extendable to other coding theory problems. In particular, to construct list-decodable codes, we defined a generalization of coloring, and constructed list-decodable codes with constant list size for a constant number of errors with redundancy less than twice the GV bound based on an existential result of a combinatorial object. We show that the probabilistic construction of these objects succeeds with overwhelming probability. Our code construction is more flexible and can be used to correct edits.   To the best of our knowledge, this is the only general construction method for list decoding with constant list size for a constant number of errors with redundancy less than twice the GV bound. Furthermore, we introduced the problem of incremental synchronization when there is uncertainty about the distance between the two versions of the data, as is almost always the case in practice. Finally, we introduced a degree reduction technique that enabled us to construct the first codes correcting a burst of $\omega(\log n)$ edits.

\ifhighlightchanges
\color{blue}
\fi
\section{appendix}
In the main body, we presented our results by assuming the color (or syndrome) $\Phi(x)$ in codes of the form $\{(x,\Phi(x)):x\in\Sigma_q^n\}$ can be transmitted without errors. Here, we remove this assumption by showing that $\Phi(x)$ can be protected with asymptotically negligible redundancy. 
\begin{lemma}\label{lem:rep1}
    Let $\cE$ be a channel applying $k$ edits to its input. Furthermore, let $G_n$ and $G_m$ be the confusion graphs for $\Sigma_q^n$ and $\Sigma_q^m$, respectively, with respect to $\cE$. Suppose $\Phi:V(G_n)\to \Sigma_q^{m}$ is a coloring of $G_n$, and $\Phi_1:V(G_m)\to \Sigma_q^{r}$ is a coloring of $G_m$. Let $\Rep_{2k+1}$ be an encoder that repeats every symbol $2k+1$ times. Then
    \[ C:=\{(x,\Phi(x),\Rep_{2k+1}(\Phi_{1}(\Phi(x)))):x\in\Sigma_q^n\}
    \] is a systematic code that can correct $k$ edits.
\end{lemma}

\begin{proof}
    Let $E:\Sigma_q^n\to \Sigma_q^{n+m+(2k+1)r}$, $E(x):=(x,\Phi(x),\Rep_{2k+1}(\Phi_{1}(\Phi(x))))$ be the encoder for this code. Let $c=E(x)$ be the codeword encoding $x$. Suppose that $y$ is obtained from $c$ through at most $k$ edits. Let $\delta = \abs{y}-\abs{c}=\abs{y}-(n+m+(2k+1)r)$ be the \emph{net} number of insertions. Then $\hat x := y_{[1,n+\delta]}$ is obtained from $x$ through at most $k$ edits, and $\hat\Phi := y_{[n+1,n+m+\delta]}$ is obtained from $\Phi(x)$ through at most $k$ edits, and $y_{[n+m+1,\abs{y}]}$ is obtained from $\Rep_{2k+1}(\Phi_{1}(\Phi(x)))$ through at most $k$ edits. Since a repetition of $2k+1$ times can correct $k$ edits, we can get $\Phi_{1}(\Phi(x))$ from $y_{[n+m+1,\abs{y}]}$. Since $\hat\Phi$ is obtained from $\Phi(x)$ through at most $k$ edits, we can get $\Phi(x)$ from $\hat\Phi$ and $\Phi_{1}(\Phi(x))$. Since $\hat x$ is obtained from $x$ through at most $k$ edits, we can get $x$ from $\hat x$ and $\Phi(x)$. Hence, we can recover $x$ given $y$.
\end{proof}

The code and the proof are similar for list-decodable codes. We state the following lemma for completeness.
\begin{lemma}\label{lem:rep1l}
    Let $\cE$ be a channel applying $k$ edits to its input. Let $H_n$ and $G_n$ be the confusion hypergraph and graph over vertex set $\Sigma_q^n$ with respect to $\cE$, respectively. Similarly, define $H_m$ and $G_m$ over vertex set $\Sigma_q^m$. Suppose $\Phi:V(H_n)\to \Sigma_q^{m}$ is an $\ell$-labeling of $H_n$ and $\Phi_1:V(G_m)\to \Sigma_q^{r}$ is a coloring of $G_m$. Let $\Rep_{2k+1}$ be an encoder that repeats every symbol $2k+1$ times. Then 
    $$C:=(x,\Phi(x),\Rep_{2k+1}(\Phi_{1}(\Phi(x))))$$ 
    is a systematic code that can correct $k$ edits with list size $\ell$.
\end{lemma}

\begin{proof}
    The proof is similar to that of the previous lemma, except that in the last step, we recover a list of size $\ell$ that contains $x$ from $\hat x = y_{[1,n+\delta]}$ and $\Phi(x)$. 
\end{proof}

Now we apply Lemma~\ref{lem:rep1} to Corollary~\ref{cor:kedit} to get the following corollary on codes correcting $k$ edits.
\begin{corollary}
    There is an encoder $E:\Sigma_{q(n)}^n\to \Sigma_{q(n)}^{n+4k\log_{q(n)}n+4k+O(\log_{q(n)}\log n)}$ that can correct $k$ edits with redundancy $4k\log n+4k\log q(n)+O(\log\log n)$ bits.
\end{corollary}
\begin{proof}
We use the notation in Lemma~\ref{lem:rep1}.
    By Corollary~\ref{cor:kedit}, we can construct a coloring $\Phi:G_n\to \Sigma_{q(n)}^{4k\log_{q(n)}n+4k+O(\log_{q(n)}\log n)}$. Furthermore, by Corollary~\ref{cor:kedit}, we can choose $m=4k\log_{q(n)}n+4k+O(\log_{q(n)}\log n)$ and $r=O(\log_{q(n)}\log_{q(n)}n)$ in Lemma~\ref{lem:rep1}. So $m+(2k+1)r=4k\log_{q(n)}n+4k+O(\log_{q(n)}\log n)$. The redundancy is $4k\log n+4k\log q(n)+O(\log\log n)$ bits.
\end{proof}
We similarly have the following corollary on list-decodable codes correcting $k$ edits with list size $\ell$.
\begin{corollary}
    There is an encoder $E:\Sigma_{q(n)}^n\to \Sigma_{q(n)}^{n+(3+1/\ell)k\log_{q(n)}n+(3+1/\ell)k+O(\log_{q(n)}\log n)}$ that can correct $k$ edits with list size $\ell$ with redundancy $(3+1/\ell)k\log n+(3+1/\ell)k\log q(n)+O(\log\log n)$ bits.
\end{corollary}
\begin{proof}
We use the notation in Lemma~\ref{lem:rep1l}.
    By Corollary~\ref{cor:keditl}, we can construct an $\ell$-labeling $$\Phi:H_n\to \Sigma_{q(n)}^{(3+1/\ell)k\log_{q(n)}n+(3+1/\ell)k+O(\log_{q(n)}\log n)}.$$ Furthermore, by Corollary~\ref{cor:keditl}, we can choose $m=(3+1/\ell)k\log_{q(n)}n+(3+1/\ell)k+O(\log_{q(n)}\log n)$ and $r=O(\log_{q(n)}\log_{q(n)}n)$ in Lemma~\ref{lem:rep1l}. So $m+(2k+1)r=(3+1/\ell)k\log_{q(n)}n+(3+1/\ell)k+O(\log_{q(n)}\log n)$. The redundancy is $(3+1/\ell)k\log n+(3+1/\ell)k\log q(n)+O(\log\log n)$ bits.
\end{proof}

\ifhighlightchanges
\color{black}
\fi

\bibliographystyle{abbrv}
\bibliography{bibliofile}

%\appendix
%\input{appendix.tex}
\end{document}